\documentclass[12pt, draftclsnofoot, onecolumn]{IEEEtran}
\pdfoutput=0
\usepackage{url}
\usepackage{moresize}
\usepackage{cite}
\usepackage{psfrag}
\usepackage{graphicx}
\usepackage{amsmath}
\usepackage{amsfonts}
\usepackage{amssymb}
\usepackage{amsthm}
\usepackage{multirow}
\usepackage[caption=false,font=footnotesize]{subfig}
\usepackage{tikz}
\usetikzlibrary{shapes,arrows}
\usepackage{pgfplots}
\pgfplotsset{compat=1.5.1}
\usepackage{anyfontsize}
\usepackage{makecell}

\usepackage[flushleft]{threeparttable}

\DeclareMathOperator*{\argmax}{arg\,max}
\DeclareMathOperator{\RMSE}{RMSE}
\def\minimize{\qopname\relax m{minimize}}
\newcommand{\N}{\operatorname{\mathbb{N}}}

\newcommand{\R}{\operatorname{\mathbb{R}}}

\newcommand{\E}{\operatorname{E}}
\newcommand{\Var}{\operatorname{Var}}

\newcommand{\Po}{\operatorname{Po}}
\newcommand{\Prob}{\operatorname{Pr}}
\newcommand{\euler}{e}
\newcommand{\ramuno}{i}
	
\newcommand{\abs}[1]{\lvert#1\rvert}	
\theoremstyle{remark}\newtheorem{remark}{Remark}
\theoremstyle{plain}\newtheorem{theorem}{Theorem}
\theoremstyle{plain}\newtheorem{lemma}{Lemma}
\theoremstyle{plain}
\theoremstyle{definition}\newtheorem{example}{Example}
\theoremstyle{definition}\newtheorem*{example*}{Example}
\theoremstyle{definition}\newtheorem*{notation*}{Notation}
\theoremstyle{remark}
\theoremstyle{definition}\newtheorem*{solution*}{Solution}
\theoremstyle{definition}
\let\oldmarginpar\marginpar
\renewcommand\marginpar[1]{\-\oldmarginpar[\raggedleft\footnotesize #1]%
{\raggedright\footnotesize #1}}

\DeclareMathOperator{\Exp}{Exp}
\DeclareMathOperator{\Lognormal}{LN}
\DeclareMathOperator{\TLN}{TLN}

\newcommand{\rmi}{^{-1}}
\newcommand{\rmh}{^{\dag}}
\newcommand{\rmt}{^{\operatorname{T}}}
\newcommand{\rmc}{^{*}}
\renewcommand{\vec}[1]{\boldsymbol{#1}}
\newcommand{\op}[1]{{\operatorname{#1}}}
\DeclareMathOperator{\crlb}{CRLB}
\DeclareMathOperator{\dist}{d}
\newcommand{\ZeroToInfty}{[0,\infty)}

\begin{document}
\bstctlcite{IEEEexample:BSTcontrol}
\title{\Huge Massive {MIMO} Extensions to the\vspace*{-.7cm}\newline {COST}~2100 Channel Model:\vspace*{-.7cm}\newline Modeling and Validation}%
\author{Jose~Flordelis,~\IEEEmembership{Student Member,~IEEE,} Xuhong Li,~\IEEEmembership{Student Member,~IEEE,} Ove~Edfors,~\IEEEmembership{Senior Member,~IEEE,} and Fredrik~Tufvesson,~\IEEEmembership{Fellow,~IEEE}\thanks{This work was supported by the Seventh Framework Programme (FP7) of the European Union under grant agreement no. 619086 (MAMMOET), ELLIIT---an Excellence Center at Link{\"o}ping-Lund in Information Technology, the Swedish Research Council (VR), and the Swedish Foundation for Strategic Research (SSF). Some initial work has been presented in~\cite{Gao:2015:Extension} and MAMMOET D1.2~\cite{MAMMOET:D1.2:2015}.}\thanks{Department of Electrical and Information Technology, Lund University, SE-221~00 Lund, Sweden (e-mail: \emph{firstname}.\emph{lastname}@eit.lth.se).}}%
\maketitle
\vspace*{-1.6cm}\begin{abstract}To enable realistic studies of massive multiple-input multiple-output systems, the COST \nolinebreak[1]2100 channel model is extended based on measurements. First, the concept of a base station-side visibility region (BS-VR) is proposed to model the appearance and disappearance of clusters when using a physically-large array. We find that BS-VR lifetimes are exponentially distributed, and that the number of BS-VRs is Poisson distributed with intensity proportional to the sum of the array length and the mean lifetime. Simulations suggest that under certain conditions longer lifetimes can help decorrelating closely-located users. Second, the concept of a multipath component visibility region (MPC-VR) is proposed to model birth-death processes of individual MPCs at the mobile station side. We find that both MPC lifetimes and MPC-VR radii are lognormally distributed. Simulations suggest that unless MPC-VRs are applied the channel condition number is overestimated. Key statistical properties of the proposed extensions, e.g., autocorrelation functions, maximum likelihood estimators, and Cramer-Rao bounds, are derived and analyzed.\end{abstract}

\begin{IEEEkeywords}Massive MIMO, channel measurements, channel model, large arrays, closely-located users, non-stationarity, birth-death process.\end{IEEEkeywords}%COST \nolinebreak[1]2100, 
\renewcommand{\baselinestretch}{1.4}
\section{Introduction}

\IEEEPARstart{M}{assive} multiple-input multiple-output (MIMO), MaMi for short, has since its inception~\cite{Marzetta:2010:massive} attracted considerable attention in the wireless communication community~\cite{Rusek:2013:massive,Marzetta:2016:redbook,SIG-093}. Within the last few years, an abundant body of theoretical~\cite{Marzetta:2016:redbook,SIG-093} and experimental~\cite{Gao:2015:MAMI,MalkowskyAccess17} research has shown that MaMi systems can improve the energy and spectral efficiencies of today's wireless communication systems by one to two orders of magnitude. For this reason, MaMi is considered a crucial component of the new radio (NR) air interface of the fifth generation (5G) wireless communication standard~\cite{Zhang:2018:MaMiNR}.

As is well known, the radio propagation channel ultimately dictates the performance of any wireless communication system. The availability of sufficiently accurate propagation channel models is therefore of critical importance to the design and evaluation of new wireless systems. The 3GPP SCM-3D~\cite{SCM-3D}, the WINNER II/WINNER+~\cite{WINNERII:2008,WINNERPlus:2010}, and the COST \nolinebreak[1]2100~\cite{VerdoneCOST12,Cardona:2016:COST} are examples of channel models customarily used in the development and validation of 5G networks. Among them, the Quasi Deterministic Radio Channel Generator (QuaDRiGa)~\cite{Jaeckel:2014:QuaDRiGa,QuaDRiGa}, an extension of the 3GPP-3D and WINNER II/WINNER+ channel models, is especially popular because of its enhanced spatiotemporal consistency. Moreover, it possesses certain features that enable MaMi simulations, namely support for spherical wavefronts, dispersion of clusters in elevation, and independence of arrival/departure angles~\cite{QuaDRiGa}.

Experiment shows~\cite{Payami:2012:LSF,Gao:2013:asilomar,Flordelis:acc:18} that at least two additional aspects of the MaMi propagation channel need to be considered. The first of them~\cite{Payami:2012:LSF,Gao:2013:asilomar} is the presence at the base station (BS) of non-stationarities caused by the appearance and disappearance of clusters along physically-large arrays (PLAs). Such BS-side non-stationarities have been addressed in theoretical MaMi channel models~\cite{Wu:2014:MaMiModel,Lopez:2018:MaMiModel} in which the appearances and disappearances of clusters are typically modeled as Markov birth-death processes obeying certain survival probabilities. However, multiuser consistency becomes problematic in such models, because the Markov birth-death processes are not linked to the geometry of the environment. Whereas the original QuaDRiGa model neither considers BS-side non-stationarities, nor ensures multiuser consistency, some progress has been made by Oliveras {\sl et al.}~\cite{Oliveras:2016:pimrc} to remedy these difficulties. However, the novel ideas presented in~\cite{Oliveras:2016:pimrc} are yet to be validated by experiments, and integrated in the mainstream QuaDRiGa model.

Secondly, it appears~\cite{Chong:lifetime:05,Salmi:rimax:09,Zentner:2010:lifetime,Wang:2012:lifetime,Zhu:2015:tracking,He:2015:lifetime_model,Wang:2015:lifetime,Mahler:2016:lifetime,Li:2018:ekf} that, just as clusters, individual multipath components (MPCs) undergo birth-death processes. The MPC lifetimes when tracked along a straight-line segment at the mobile station (MS) side have been found to obey certain probability distributions, e.g., the exponential distribution~\cite{Chong:lifetime:05,He:2015:lifetime_model,Li:2018:ekf}, the lognormal distribution~\cite{Wang:2012:lifetime,Wang:2015:lifetime,Mahler:2016:lifetime,Li:2018:ekf}, and the Birnbaum-Sanders distribution~\cite{Mahler:2016:lifetime,Li:2018:ekf}. More important, closely-located users only a few wavelengths apart can encounter different MPCs~\cite{Flordelis:acc:18}. Channel models that incorporate birth-death processes at the level of individual MPCs are of interest to emerging applications such as MaMi communications~\cite{Adhikary:2013:JSDM,Junyoung:2014:JSDM}, vehicle-to-vehicle communications~\cite{He:2015:lifetime_model,Mahler:2016:lifetime}, or multipath-assisted localization~\cite{Wang:2012:lifetime,Zhu:2015:tracking,Li:2018:ekf,Leitinger:2018:rfslam}, which can use this information to improve their performance. To the best of our knowledge, none of the earlier referred 5G channel models currently considers birth-death processes for individual MPCs.

In this paper, MaMi extensions of the COST \nolinebreak[1]2100 channel model are presented that describe BS-side non-stationarities of the clusters and MS-side birth-death processes of individual MPCs, introduced above. The main contributions of the paper are the following:
\begin{itemize}
\item The notion of BS-side visibility region (BS-VR) is proposed to model BS-side non-sta\-tion\-ar\-i\-ties. Measured outdoor MaMi channels show that the number of BS-VRs is Poisson distributed and their lifetimes obey an exponential law.
\item The maximum likelihood estimator and the Cramer-Rao bound of the parameters of the BS-VRs are rigorously derived. A closed-form expression of the autocorrelation function (ACF) of the number of the number of observed BS-VRs is provided.
\item The influence of the choice of BS-VR parameters on the MaMi propagation channel is investigated through simulations. We find that under certain conditions longer BS-VRs can help decorrelating closely-located users.
\item The notions of MPC-VR and MPC gain function are proposed to model birth-death processes of individual MPCs at the MS side. Measured indoor MaMi channels reveal that MPC lifetimes and MPC-VR radii are lognormally distributed.
\item The small-scale fading ACF with the MPC gain function is always smaller than without it. When the gain function is applied to a group of closely-located users their condition number decreases by a few dB and comes to agree with measurements.
\end{itemize}
The present work extends our preliminary studies in~\cite{Gao:2015:Extension} with further analyses, revised insights, and new measurements to characterize the MPC gain function. A MATLAB implementation of the COST \nolinebreak[1]2100 model with the proposed extensions is freely available at~\cite{COST2100:GitHub}.

\section{Notation}
Throughout the paper, boldface lowercase letters represent column vectors, boldface uppercase letters matrices, and calligraphic letters sets. Thus, $\vec{I}$ denotes the identity matrix, $\vec{A}\rmt$ the transpose, $\vec{A}\rmh$ the Hermitian transpose, $\ramuno$ the imaginary unit, $\R^+$ the set $\ZeroToInfty,$ $\Prob(A)$ the probability of event $A,$ $1_A(\cdot)$ the indicator function of $A$, $\E(\cdot)$ the expectation operator, and $\Phi(x)$ the CDF of $\mathcal{N}(0,1).$

\section{The {COST}~2100 {MIMO} Channel Model}\label{sec_COST2100_intro}

We briefly review some concepts and terminology of the COST \nolinebreak[1]2100 model. For further details, the reader is referred to~\cite{VerdoneCOST12,Cardona:2016:COST}. The COST \nolinebreak[1]2100 model falls into the category of geometry-based stochastic channel models (GSCM)~\cite{MolischETT03}, which are built around the notion of a \emph{cluster}, i.e., a group of MPCs with similar delay and angular parameters~\cite{CzinkThesis07}. Clusters model interactions of the transmit signal with scattering objects in the environment such as building facades, trees and street furniture, in outdoor environments, or inner walls, pillars and office equipment, in indoor settings. Interactions happen at so-called \emph{scattering points}. In GSCMs, the locations of clusters and scattering points are drawn randomly from prescribed multivariate distributions depending on the simulated environment. MPCs are then obtained by mere geometric ray tracing of the transmit signal through the scattering points~\cite{VerdoneCOST12}.

The distinguishing feature of the COST \nolinebreak[1]2100 model is perhaps the use of so-called visibility regions (VRs) to model non-stationarities in a spatiotemporally consistent way. When a MS moves inside a VR, the associated cluster is active and its MPCs contribute to the MS-BS double-directional impulse response. Similarly, when the MS moves out of a VR, the associated cluster is no longer visible and it does not contribute. The situation is depicted in Fig.~\ref{fig_MSVR}.
The two MaMi channel aspects alluded to in the introduction are also shown. The first aspect, illustrated in Fig.~\ref{fig_BSVR}, has to do with the presence of non-stationarities at the BS side in MaMi settings with PLAs~\cite{Payami:2012:LSF,Gao:2013:asilomar}; the second one, illustrated in Fig.~\ref{fig_MPCVR}, concerns the realization that individual MPCs undergo birth-death processes at the MS side~\cite{Chong:lifetime:05,Salmi:rimax:09,Zentner:2010:lifetime,Wang:2012:lifetime,Zhu:2015:tracking,He:2015:lifetime_model,Wang:2015:lifetime,Mahler:2016:lifetime,Li:2018:ekf}. To model these phenomena, the COST \nolinebreak[1]2100 channel model is extended with the concepts of BS-VRs, in Sec.~\ref{sec_large}, and MPC-VRs and MPC gain functions, in Sec.~\ref{sec_extension_2}.

\begin{figure*}[!ht]
  \psfrag{A}[][]{\small{A}}
  \psfrag{B}[][]{\small{B}}
  \psfrag{C}[][]{\small{C}}
  \psfrag{D}[][]{\small{D}}
  \psfrag{BS}[][]{\small{BS}}
  \psfrag{msroute}[][]{\small{MS route}}
  \centering
  \hspace*{\fill}
  \subfloat[]{\centering\label{fig_MSVR}\includegraphics[width=.31\textwidth]{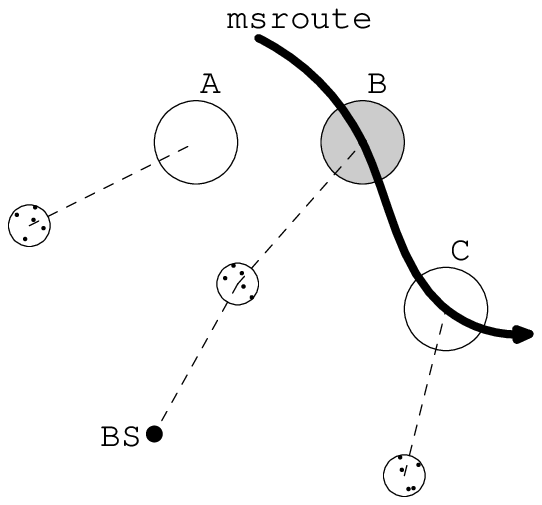}}\hfill
  \subfloat[]{\centering\label{fig_BSVR}\includegraphics[width=.31\textwidth]{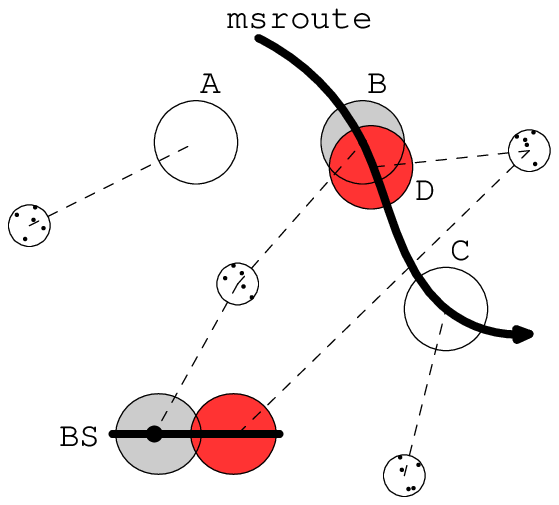}}\hfill
  \subfloat[]{\centering\label{fig_MPCVR}\includegraphics[width=.31\textwidth]{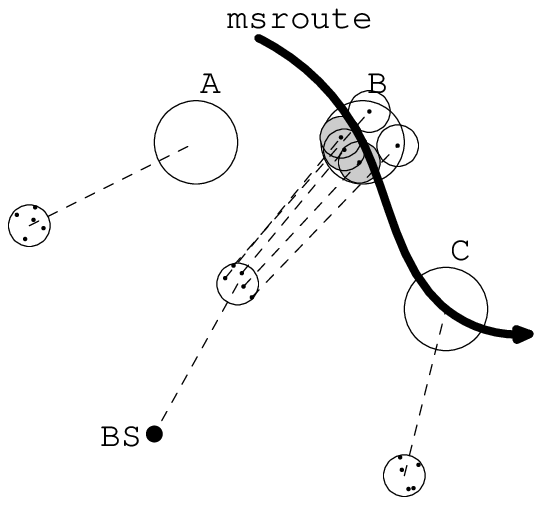}}
  \hspace*{\fill}
  \caption{A MS moves in and out of various VRs, labeled A, B, C, and D. (a) When the MS enters a VR, the associated cluster becomes active (VR B). (b) When a PLA is used, different parts can observe different sets of clusters (VR B, or VR D, or VRs B and D), as determined by BS-VRs. (c) The relative gain of individual MPCs is controlled by MPC-VRs, one for each MPC.}
  \label{fig_VR}
\end{figure*}

\section{Extension for Physically-Large Arrays}\label{sec_large}

\subsection{Measurements and Data Processing}\label{sec_bsvr_data}

Channels used in this section were measured at the parking lot outside the E-building of the Faculty of Engineering of Lund University, Lund, Sweden. At the BS side, a 7.5~m-long, 128-element virtual uniform linear array (ULA) was mounted on a rooftop, three floors above the ground level. Eight MS sites, four in line-of-sight (LOS) propagation conditions to the BS and four in non-LOS (NLOS), were considered. At each site, the radio channel between the ULA and a single-antenna MS was measured at five locations interspaced by 0.5~m. Measurements were acquired at a carrier frequency of 2.6~GHz, and over a bandwidth of 50~MHz. All antennas were vertically-polarized and omnidirectional on the horizontal plane.

To extract clusters of MPCs, the measured channels were subsequently processed in the following way. First, the ULA was partitioned into sets of neighboring antennas by using a 10 element-long sliding window. Then, MPC parameters (i.e., delay, azimuth angle-of-arrival, and complex amplitudes) were estimated by applying the SAGE~\cite{Fleury:1999:SAGE} super-resolution channel estimation algorithm to each window. Finally, clusters were identified using the KPowerMeans~\cite{CzinkThesis07} joint clustering and tracking algorithm. Clusters containing less than 2.5\% of the total power in any window, or surviving less than 5 windows, were discarded. See~\cite{Gao:2015:MAMI} for further details.

\subsection{System Model and Assumptions}\label{sec_model}

Because of the nature of the measurements available, the scope of this section is limited to ULAs. However, the results obtained below can be extended to the case of two-dimensional arrays by applying, e.g., the techniques demonstrated in Sec.~\ref{dist_radii}. Such an extension is best addressed with complementary measurements and therefore we leave it for future work.

Let the ULA be located at $0\leq x_1\le x\le x_2$, and assume that BS-VRs can be modeled as arising from a birth-death process having the following properties:
\begin{enumerate}\item[A1.] Birth events (i) are independent from each other; (ii) for positive increments~$h$, $\Prob($exactly one birth event during$\ (x,x+h]) = \lambda h + o(h)$ as $h\to 0$ for some $\lambda>0$ called the birth rate; and (iii) $\Prob($at least two birth events during$\ (x,x+h])=o(h)$ as $h\to 0$. By~\cite[Chapter 8, Definition II]{Gut:2009:P}, this is equivalent to the assertion that the birth count process $\{N_\op{birth}(x),x\ge 0\}$, where
      \begin{align}\label{eq_N_birth}
        N_\op{birth}(x) = \text{ \# births events in } (0,x],
      \end{align}
      is a Poisson process with intensity~$\lambda>0$.
\item[A2.] The \emph{true} lifetimes of BS-VRs are (i) independent copies of some nonnegative random variable~$Y$ with probability density function (PDF) $f_Y:\R^+\to\R^+$ with finite first moment; (ii) independent of birth events.
\item[A3.] In experiments one can only observe (perhaps partially) BS-VRs that overlap with~$[x_1,x_2]$. Hence, the outcome of any experiment is a sequence of intervals $[a_1,b_1],\ldots,[a_n,b_n]\subseteq[x_1,x_2]$, where $a_1\leq\ldots\leq a_n$ are the \emph{observed} birth positions, $\upsilon_1=b_1-a_1,\ldots,\upsilon_n=b_n-a_n$ the \emph{observed} lifetimes, and $n$ the number of observed BS-VRs.
\end{enumerate}

\subsection{Distribution of the Number of BS-VRs}\label{sec_num_bsvr}
Let the random variable $N(x_1,x_2)$ denote the number of BS-VRs observed in the interval $[x_1,x_2].$ Then, some natural questions to ask are: What is the expected number of observed BS-VRs, $\E(N(x_1,x_2))$? And, how does this number depend on the interval length~$x_2-x_1\ge 0$? To answer these questions, in this section we derive the distribution function of $N(x_1,x_2)$. We start by noting that the total number of observed BS-VRs can be written as the sum of two independent contributions:
\begin{itemize}
\item The number, $N_\op{new}(x_1,x_2)$, of BS-VRs ``born'' in the interval $(x_1,x_2]$. Clearly, this quantity can be obtained by the difference
  \begin{align*}
    N_\op{new}(x_1,x_2) = N_\op{birth}(x_2)-N_\op{birth}(x_1).  
  \end{align*}
  Since $N_\op{birth}(x)$ is, by hypothesis, a Poisson process, it follows that
  \begin{align}\label{eq_N_new}
    N_\op{new}(x_1,x_2) \in \Po(\lambda\cdot (x_2-x_1)),  
  \end{align}
  where $\Po(m)$ denotes the Poisson distribution with parameter $m\ge 0$ and probability mass function (PMF)
  \begin{align}
    p(n) = \euler^{-m}\frac{m^n}{n!},\quad n=0,1,2,\ldots,
  \end{align}
  where one should assign the value 1 to $0^0$.
\item The number, $N_\op{alive}(x_1)$, of BS-VRs born before, and still alive at, $x=x_1.$ The distribution of $N_\op{alive}(x)$ for $x_1$ sufficiently far away from $x=0$ is given as part of the proof of Theorem~\ref{theo_N}.
\end{itemize}
The following result is a key ingredient of our proposed channel model extension for PLAs.

\begin{theorem}[Number of observed BS-VRs for PLAs]\label{theo_N}Let the assumptions in Sec.~\ref{sec_model} hold. Then, the number of observed BS-VRs in the interval~$[x_1,x_2]$ has the distribution
    \begin{align}\label{eq_N_total}
      N(x_1,x_2) \in \Po(\lambda\cdot(x_2-x_1) + \lambda\cdot\E(Y)).
    \end{align}
\end{theorem}
\begin{proof}By (A2) and the fact that the sum of two independent Poisson-distributed random variables is also Poisson-distributed, it suffices to show that $N_\op{alive}(x_1)\in\Po(\lambda\cdot\E(Y)).$

  Clearly, only birth events $N_\op{birth}(x_1)$ that lead to BS-VRs extending past $x=x_1$ contribute to $N_\op{alive}(x_1).$ Let $0\leq u<x_1.$ Then, BS-VRs emanating from a neighborhood $(u,u+h]$ of~$u$ are observed at $x=x_1$ with probability
    \begin{align}
      P(Y\ge x_1-u) = 1-F_Y(x_1-u).  
    \end{align}
    Since $\{N_\op{birth}(x),x\ge 0\}$ has independent increments and BS-VR lifetimes are generated independently of each other, it follows that the surviving BS-VRs may be interpreted as arising from a non-homogeneous Poisson process $\{N'_\op{birth}(x), x\ge 0\}$ with position-dependent intensity $\lambda(x)$ given by
  \begin{align}\label{eq_lambda}
    \lambda(x) = \lambda\cdot P(Y\ge x_1-x).
  \end{align}
  From the theory of non-homogeneous Poisson processes~\cite{Gut:2009:P} we obtain
  \begin{align*}
    N_\op{alive}(x_1) = N'_\op{birth}(x_1)-N'_\op{birth}(0) = \Po(\lambda\cdot\int_0^{x_1} P(Y\ge x_1-u)\,du).
  \end{align*}
  Finally, assuming that the starting point of the array can be chosen as far from $x=0$ as desired produces
  \begin{align}\label{eq_N_alive}
    N_\op{alive}(x_1)
    = \lim_{x_1\to\infty} \Po(\lambda\cdot\int_0^{x_1} P(Y\ge u)\,du)
    = \Po(\lambda\cdot\E(Y)),
  \end{align}
  which completes the proof.  
\end{proof}

\begin{remark}Note that Theorem~\ref{theo_N} gives us the answers we seek, namely that for PLAs, the expected number of observed BS-VRs,~$\E(N(x_1,x_2))$, depends on the environment through (i) the intensity $\lambda>0$ of the underlying Poisson process, and (ii) the mean $\E(Y)$ of the true BS-VR lifetimes; and on the system through (iii) the length $L=x_2-x_1$ of the PLA. In particular, note that the specific shape of the distribution $f_Y:\R^+\to\R^+$ does not matter.\end{remark}

\begin{remark}In this paper we are interested mainly in extensions of the COST \nolinebreak[1]2100 model supporting MaMi, but it is also instructive to compare these with the ``conventional'' COST \nolinebreak[1]2100 model, which only considers physically-compact arrays (PCAs). Observe that
  \begin{align}
    N(x_1,x_2) \to N_\op{alive}(x)\text{ as }x_1,x_2\to x,    
  \end{align}
  so that the number of VRs observed by a PCA placed anywhere along the line of the PLA is
  \begin{align}\label{eq_N_total_compact}
    N(x) \in \Po(\lambda\cdot\E(Y)).
  \end{align}
  As a matter of fact, the conventional COST \nolinebreak[1]2100 models the number of \emph{visible far clusters} (and thus of VRs) as being Poisson-distributed with a certain prescribed intensity which depends on the simulated environment; see~\cite[Chapter 3]{VerdoneCOST12}. Thus, our proposed channel model extension for PLAs naturally contains and generalizes the conventional COST \nolinebreak[1]2100 model.
\end{remark}

Fig.~\ref{fig_num_cl} verifies Theorem~\ref{theo_N} empirically. In particular, the empirical cumulative density functions (ECDFs) of~$N_\op{alive}(x_1)$, $N_\op{new}(x_1,x_2)$, and~$N(x_1,x_2)$ for the measured LOS sites described in Sec.~\ref{sec_bsvr_data} are presented, along with Poisson fits. Empirical curves are based on 20 samples, five from each measured LOS site. A quick inspection reveals that these quantities are well approximated by Poisson-distributed random variables; the Chi-square goodness-of-fit test at significance level~$\alpha=0.05$ statistically confirms this visual impression. Empirical evidence from measured NLOS sites (not shown due to lack of space) also supports the hypothesis that the number of observed BS-VRs is well modeled by~\eqref{eq_N_new}, \eqref{eq_N_total}, and~\eqref{eq_N_alive}.
\begin{figure}[!t]
  \centering
  \resizebox{.65\textwidth}{!}{\centering\input{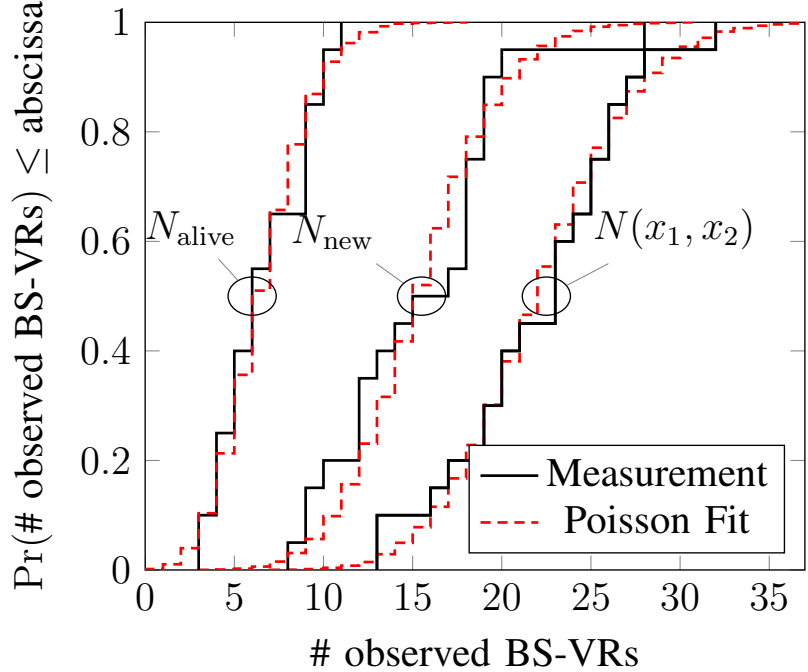}}
  \caption{ECDFs of~$N_\op{alive}(x_1)$, $N_\op{new}(x_1,x_2)$, and~$N(x_1,x_2)$ for measured LOS sites, and Poisson fits, with $x_2-x_1\approx 7.5$~m, $\Delta_0\approx 0.23$~m.}
  \label{fig_num_cl}
\end{figure}

\subsection{Maximum-Likelihood Estimation of BS-VR Lifetimes}\label{ml_est}

Let~$\mathcal{I}=\{1,\ldots,n\},$ and define
\begin{align}\label{eqsub:main}
  \mathcal{N}_{00} = &\left\{i\in\mathcal{I}:x_1<a_i,b_i<x_2\right\},\quad \mathcal{N}_{01} = \left\{i\in\mathcal{I}:x_1<a_i,b_i=x_2\right\},\nonumber\\
  \mathcal{N}_{10} = &\left\{i\in\mathcal{I}:x_1=a_i,b_i<x_2\right\},\quad \mathcal{N}_{11} = \left\{i\in\mathcal{I}:x_1=a_i,b_i=x_2\right\},
\end{align}
to be the subsets of indices whose BS-VRs are either observed completely, extend into $x>x_2$, into $x<x_1$, or into both $x<x_1,$ $x>x_2$, respectively. Further, let $n_{00}$, $n_{01}$, $n_{10}$, $n_{11}$ denote the number of elements in $\mathcal{N}_{00}$, $\mathcal{N}_{01}$, $\mathcal{N}_{10},$ $\mathcal{N}_{11}$ so that
\begin{align}\label{eq_conservation_n}
  n = n_{00} + n_{01} + n_{10} + n_{11}
\end{align}
holds. It is shown in~\cite{Flordelis19} that, under the assumptions Sec.~\ref{sec_model}, the likelihood function that makes use of all the available information has the form
\begin{align}\label{eq_like}
  \mathcal{L}(\vec{\theta};\vec{x}) &= \frac{\lambda^n}{n!}\euler^{-\lambda(L+\E(Y_{>\Delta_0})-2\Delta_0)\int_{\Delta_0}^\infty f_Y(t;\vec{\theta}_Y)\,dt}\nonumber\\
  &\times \prod_{i\in\mathcal{N}_{00}} f_Y(\upsilon_i;\vec{\theta}_Y) \times \prod_{j\in\mathcal{N}_{01}\cup\mathcal{N}_{10}} \int_{\upsilon_j}^\infty f_Y(t;\vec{\theta}_Y)\,dt\nonumber\\
  &\times \left(\int_L^\infty (t-L)f_Y(t;\vec{\theta}_Y)\,dt\right)^{n_{11}},
\end{align}
where $\vec{\theta} = [\lambda,\vec{\theta}_Y\rmt]\rmt$~is the $(p+1)$-dimensional vector of unknown deterministic parameters to be estimated, $\vec{\theta}_Y$ is the $p$-dimensional vector parametrizing the PDF $f_Y(y;\vec{\theta}_Y)$, $\vec{x} = [a_1,b_1,\ldots,b_n]^{\op{T}}$ is a $2n$-dimensional vector of data,~$\Delta_0\ge 0$ is the minimum feature size, and~$Y_{>\Delta_0}$ is the restriction of $Y$ to $\{\omega:Y(\omega)>\Delta_0\}$. The maximum likelihood estimator (MLE)~$\hat{\vec{\theta}}$ of $\vec{\theta}$ given~$\vec{x}$ is
\begin{align}\label{eq_mle}
  (\hat{\lambda},\hat{\vec{\theta}}_Y) &= \argmax_{\vec{\theta}} \mathcal{L}(\vec{\theta};\vec{x}).
\end{align}
\begin{figure}[t]
  \centering
  \resizebox{.65\textwidth}{!}{\centering\input{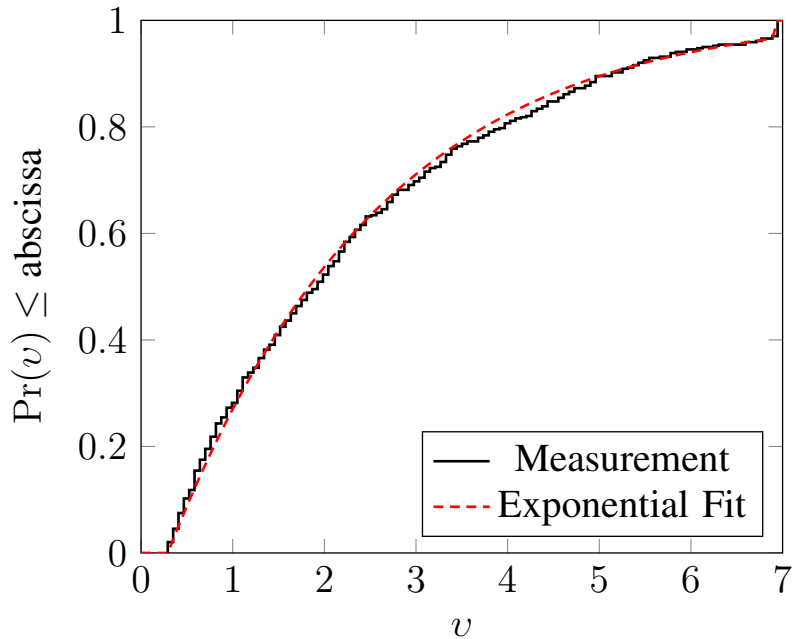}}
  \caption{ECDF of the observed lifetimes~$\upsilon_1\ldots,\upsilon_n$ of the BS-VRs for the measured LOS sites, and corresponding exponential fit based on~\eqref{eq_lbs_hat}.}
  \label{fig_exp_fit}
\end{figure}
For arbitrary PDFs $f_Y(y;\vec{\theta}_Y):\R^+\to\R^+$, closed-form solutions to~\eqref{eq_mle} are in general not available and, as a result, one has to resort to numerical methods. It turns out, however, that the true BS-VR lifetimes~$y_1,y_2,\ldots$ in our channel measurements can be well approximated by an exponential law~$\Exp(L_\op{BS})$, as indicated by Fig.~\ref{fig_exp_fit}.\footnote{Incidentally, an exponential law is assumed in~\cite{Wu:2014:MaMiModel,Lopez:2018:MaMiModel} to model the appearance and disappearance of clusters on both the array and time axes, but no empirical evidence is given to support this assumption.} That is, the random variable $Y$ has the distribution
\begin{align}\label{eq_exp_pdf}
  f_Y(y) =
  \begin{cases}
    \frac{1}{L_\op{BS}}\euler^{-\frac{y}{L_\op{BS}}},\quad&\text{for }y\ge 0;\\
    0,\quad&\text{otherwise,}
  \end{cases}
\end{align}
where $L_\op{BS}$ is the mean BS-VR lifetime. In this case, the MLE takes on a relatively simple form.

\begin{theorem}[MLE for PLAs]\label{theo_mle}Consider a birth-death process with intensity $\lambda>0$ obeying the as\-sump\-tions in Sec.\ref{sec_model}. Let $Y\in\Exp(L_\op{BS})$ be the distribution of the true BS-VR lifetimes.~Then
  \begin{enumerate}
  \item[(i)]the quantities~$n$,~$\nu=n_{11}-n_{00}$, and~$\Lambda_0=\sum_{i=1}^{n}(\upsilon_i-\Delta_0)$ are sufficient statistics;
  \item[(ii)]the MLEs of the unknown parameters $\lambda,$ $L_\op{BS}$ are given by
  \end{enumerate}
  \begin{align}
    \hat{\lambda} &= \frac{n}{L_0+\hat{L}_\op{BS}}\euler^{\frac{\Delta_0}{L_\op{BS}}},\label{eq_lambda_hat}\\
    \hat{L}_\op{BS} &= \frac{L_0\nu+\Lambda_0}{2(n-\nu)}\left(1+\sqrt{1 + \frac{4(n-\nu)L_0\Lambda_0}{(L_0\nu+\Lambda_0)^2}}\right),\label{eq_lbs_hat}
  \end{align}    
  where $L_0=L-\Delta_0$ is the array length shortened by the minimum feature size,~$\Delta_0.$
\end{theorem}

We will need the following result.
\begin{lemma}\label{lemma_equiv}
  Let the assumptions in Theorem~\ref{theo_mle} hold, let $\Delta_0\ge 0,$ and define $\lambda_0=\lambda\euler^{-\frac{\Delta_0}{L_\op{BS}}}.$ Then, two PLAs given by the 4-tuples $(L,L_\op{BS},\lambda,\Delta_0)$ and~$(L_0,L_\op{BS},\lambda_0,0)$ have the same likelihood function, provided that observations $\upsilon_1,\ldots,\upsilon_n$ in the former one are shortened by~$\Delta_0.$
\end{lemma}
\begin{proof}
  Assume~$(L,L_\op{BS},\lambda,\Delta_0)$, apply~\eqref{eq_exp_pdf} to~\eqref{eq_like} to obtain
  \begin{align}\label{eq_like_exp}
    \mathcal{L}(\vec{\theta};\vec{x}) \propto \lambda^n \euler^{-\lambda(L-\Delta_0+L_\op{BS})\euler^{-\frac{\Delta_0}{L_\op{BS}}}} L_\op{BS}^\nu \euler^{-\sum_{i=1}^n\frac{\upsilon_i}{L_\op{BS}}},
  \end{align}
  which is unchanged under the mapping $L\mapsto L_0,$ $\lambda\mapsto\lambda_0,$ $\Delta_0\mapsto 0$ $\upsilon_i\mapsto\upsilon_i-\Delta_0.$
\end{proof}

We are now ready to prove Theorem~\ref{theo_mle}.
\begin{proof}
  By Lemma~\ref{lemma_equiv} we can assume without loss of generality that~$\Delta_0=0.$ Claim (i) follows directly from an application of the Neyman-Fisher factorization theorem. We now prove (ii). Taking the partial derivatives of the logarithm of the likelihood function~\eqref{eq_like_exp}, we have
  \begin{align*}
    \frac{\partial\log{\mathcal{L}(\vec{\theta};\vec{x})}}{\partial\lambda} &= \phantom{-}\frac{n}{\lambda} - (L+L_\op{BS}),\\
    \frac{\partial\log{\mathcal{L}(\vec{\theta};\vec{x})}}{\partial L_\op{BS}} &= -\lambda + \frac{\sum_{i=1}^n\upsilon_i}{L_\op{BS}^2} + \frac{\nu}{L_\op{BS}},
  \end{align*}
  and setting them equal to zero produces
  \begin{subequations}\label{eqsub3:main}
    \begin{align}
      0&=\phantom{-}\frac{n}{\lambda} - (L + L_\op{BS}),\label{eqsub3:eq_1}\\
      0&=-\lambda + \frac{\nu}{L_\op{BS}} + \frac{\sum_{i=1}^n\upsilon_i}{L_\op{BS}^2}.\label{eqsub3:eq_2}
    \end{align}
  \end{subequations}
  Solving~\eqref{eqsub3:eq_1} for~$\lambda$ and inserting into~\eqref{eqsub3:eq_2} produces
  \begin{align*}
    -\frac{n}{L+L_\op{BS}} + \frac{\nu}{L_\op{BS}} + \frac{\sum_{i=1}^n\upsilon_i}{L_\op{BS}^2} = 0,
  \end{align*}
  which, by grouping terms of the same degree, can be rewritten as $$aL_\op{BS}^2 + bL_\op{BS} + c = 0$$ with $a = \nu-n,$ $b = L\nu + \sum_{i=1}^n\upsilon_i,$ and $c = L(\sum_{i=1}^n\upsilon_i).$ Solving for~$L_\op{BS}$ yields the two solutions $$\hat{L}_\op{BS} = \frac{-b\pm\sqrt{b^2-4ac}}{2a}.$$ Clearly, we have $c\ge 0$ and, by~\eqref{eq_conservation_n}, we see that $a\leq 0$. It follows that $b^2-4ac\ge b^2,$ and so only the solution $$\hat{L}_\op{BS} = \frac{-b+\sqrt{b^2-4ac}}{2a}$$ pertains to the permissible range $\hat{L}_\op{BS}\ge 0$ of the parameter~$L_\op{BS}.$ Since~$a\leq 0$, it is clear that $\hat{L}_\op{BS}$ indeed maximizes the log-likelihood function. Finally, we see from~\eqref{eqsub3:eq_1} that the MLE of the parameter~$\lambda$ must be~\eqref{eq_lambda_hat}.
\end{proof}

\begin{remark}
  Note in Theorem~\ref{theo_mle} that if $n_{00}=n_{11}+\frac{\Lambda_0}{L_0}$ holds (and so $n_{11},\frac{\Lambda_0}{L_0}\in\N^+$), then \eqref{eq_lbs_hat} reduces to~$\hat{L}_\op{BS}=(\frac{L_0\Lambda_0}{n-\nu})^{\frac{1}{2}}$. On the other hand, if $n_{00}=n_{01}=n_{10}=0$, then $\hat{L}_\op{BS}=\infty.$
\end{remark}

It is a well-known result~\cite[Chapter~7]{Kay:1993:estimation} that under certain regularity assumptions on the likelihood function, the MLE is asymptotically unbiased (and thus $\E(\hat{\vec{\theta}})\to\vec{\theta}$ as $n\to\infty$) and efficient (and thus $\Var(\hat{\vec{\theta}})\to\text{CRLB}$ as $n\to\infty$ as well, where CRLB means the Cramer-Rao lower bound of the PDF~$f_{\vec{\upsilon}}(\vec{\upsilon};\vec{\theta}).$) It is therefore of interest to obtain the CRLB for the estimates of the intensity~$\lambda$ and lifetimes~$L_\op{BS}.$

\begin{theorem}[Cramer-Rao lower bound]\label{theo_crlb}Consider a birth-death process with intensity $\lambda>0$ satisfying the assumptions in Sec.\ref{sec_model}, and let $Y\in\Exp(L_\op{BS})$ be the distribution of the true lifetimes. Then, the Fisher information matrix (FIM) is
  \begin{align}
    \vec{I}(\vec{\theta}) = \lambda_0\begin{bmatrix}(L_\op{BS}+L_0)/\lambda^2&1\\1&(L_\op{BS}+L_0)/{L_\op{BS}^2}\end{bmatrix}.
  \end{align}
  The CRLB is found as
  \begin{align}
    \crlb_{\lambda} = [\vec{I}(\vec{\theta})\rmi]_{11},\quad\crlb_{L_\op{BS}} = [\vec{I}(\vec{\theta})\rmi]_{22}.\label{crlbs}
  \end{align}
\end{theorem}

The proof is a routine exercise in view of the equalities $\E(n)=\lambda_0(L_\op{BS}+L_0),$ $\E(n_{11}-n_{00})=\lambda_0(L_\op{BS}-L_0),$ $\E(\Lambda_0)=\lambda_0 L_\op{BS} L_0,$ and is thus omitted.\footnote{Even more is true. The equalities $\E(n)=\lambda(L_\op{BS}+L),$ $\E(n_{11}-n_{00})=\lambda(L_\op{BS}-L),$ $\E(\Lambda)=\lambda L_\op{BS} L$ hold for \emph{arbitrary} $f_Y:\R^+\to\R^+,$ and $\Delta_0=0.$} Numerous examples can be found in,~e.g.,~\cite{Kay:1993:estimation}. Note that the expectation is with respect to~$f_{\vec{x}}(\vec{x};\vec{\theta}).$

The well-known information inequality~\cite{Kay:1993:estimation} asserts that $\E((\hat{\vec{\theta}}-\vec{\theta})(\hat{\vec{\theta}}-\vec{\theta})\rmt)\succeq \vec{I}(\vec{\theta})\rmi$ for~$\hat{\vec{\theta}}$ unbiased, from which it follows that
\begin{align}  
  \Var\left(\frac{\hat{\lambda}}{\lambda}\right) \ge \frac{\crlb_\lambda}{\lambda^2} = \left(\frac{1}{\lambda_0 L_0}\right)\frac{1 + L_\op{BS}/L_0}{1 + 2L_\op{BS}/L_0},\label{eq_crlb_lambda}
\end{align}
and
\begin{align}
  \Var\left(\frac{\hat{L}_\op{BS}}{L_\op{BS}}\right) \ge \frac{\crlb_{L_\op{BS}}}{L_\op{BS}^2} = \left(\frac{1}{\lambda_0 L_0}\right)\frac{1 + L_\op{BS}/L_0}{1 + 2L_\op{BS}/L_0}.\label{eq_crlb_lbs}
\end{align}  
For fixed birth rate~$\lambda>0$, the normalized Cramer-Rao bounds satisfy
\begin{align}
  \frac{0.5}{\lambda_0L_0} \leq \frac{\crlb_{\lambda}}{\lambda^2} = \frac{\crlb_{L_\op{BS}}}{L_\op{BS}^2} \leq \frac{1}{\lambda_0L_0},\label{eq_crlb_tmp}
\end{align}
and so the influence of the parameter~$L_\op{BS}$ is limited. In fact,~\eqref{eq_crlb_tmp} reveals that the normalized CRLBs are dominated by the inverse of the intensity-aperture product~$\lambda_0L_0.$ Increasing the length~$L$ of the array can therefore reduce the variance of~$\hat{\vec{\theta}}$ significantly. Likewise, better performance of~$\hat{\vec{\theta}}$ should be expected in environments with many scatterers (clusters) such that~$\lambda$ is large. We also note from~\eqref{eq_crlb_tmp} that the effect of the minimum feature size~$\Delta_0$ on the CRLBs can be mostly neglected so long as $\frac{\Delta_0}{L_\op{BS}}\ll 1,$ $\frac{\Delta_0}{L}\ll 1$ hold.

\begin{figure*}[!t]
  \psfrag{LeffbyLdB}[][]{\scriptsize{$10\log_{10}(L_\op{BS}/L_0)$}}
  \psfrag{stdlambda}[][]{\scriptsize{$\RMSE(\hat{\lambda})/\lambda$}}
  \psfrag{stdleff}[][]{\scriptsize{$\RMSE(\hat{L}_\op{BS})/L_\op{BS}$}}
  \psfrag{averageratio}[][]{\scriptsize{Average~$n_{xy}$ / Average~$n$}}
  \psfrag{lambdaLeq2520151005}[][]{\tiny{$\lambda_0L_0=5,10,15,20,25$}}
  \psfrag{MMM1}[][]{\tiny{$\sqrt{\crlb}$}}
  \psfrag{MMM2}[][]{\tiny{MLE}}
  \psfrag{MMM3}[][]{\tiny{MoME}}
  \psfrag{n00}[][]{\scriptsize{$n_{00}$}}
  \psfrag{n01}[][]{\scriptsize{$n_{01}$}}
  \psfrag{n10}[][]{\scriptsize{$n_{10}$}}
  \psfrag{n11}[][]{\scriptsize{$n_{11}$}}
  \centering
  %trim={<left> <lower> <right> <upper>}
  \includegraphics[trim={0 0 0 .1cm},clip,width=\textwidth]{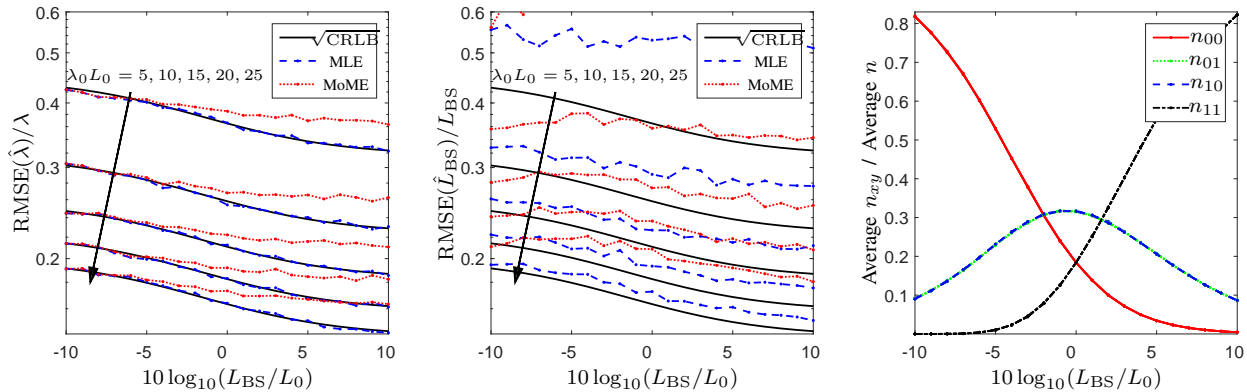}
  \caption{Plots of the RMSE and proportions of the various BS-VR types. For each point, $10,000$ experiments have been randomly generated. (Left) The MLE, MoME, and CRLB of~$\lambda.$ (Middle) Idem for $L_\op{BS}$. (Right) Radii of the average of~$n_{00},$ $n_{01},$ $n_{10},$ and~$n_{11},$ to the average of~$n.$ (Note that the plots of~$n_{01}$ and~$n_{10}$ overlap.)}
  \label{fig_mle_std}
\end{figure*}
To evaluate the extent to which the estimators~$\hat{\lambda},$ $\hat{L}_\op{BS}$ given by~\eqref{eq_lambda_hat},~\eqref{eq_lbs_hat} approach the CRLBs~\eqref{eq_crlb_lambda},~\eqref{eq_crlb_lbs} we run Monte Carlo simulations. Fig.~\ref{fig_mle_std} shows the root-mean-square error (RMSE) relative to the magnitude of the estimated parameter. As baseline for comparison we use the method-of-moments estimator (MoME)
\begin{align}
  \hat{L}_\op{BS} = \frac{T}{1-T/L_0},\quad \hat{\lambda} = \frac{n}{L_0+\hat{L}_\op{BS}},\quad T=\frac{\Lambda_0}{n}.\label{eq_lbs_mome}
\end{align}
The second equation in~\eqref{eq_lbs_mome} follows from the identity~$\E(\upsilon)=\frac{L_\op{BS}L_0}{L_\op{BS} + L_0},$ and by approximating~$\E(\upsilon)$ by $T=\frac{\Lambda_0}{n}.$ MoM-based estimators are sometimes used in the literature~\cite{Zhu:2013:COST2100}.\footnote{The derivation of~\eqref{eq_lbs_mome} uses (i) $Y\in\Exp(L_\op{BS})$, and (ii) $\E(\upsilon)=\frac{L_\op{BS}L_0}{L_\op{BS}+L_0},$ none of which is obvious. Less sophisticated estimators are possible. One can, for instance, assume $Y$ constant (and thus knowledge of type (i) is not needed), or ignore the fact that observed BS-VR lifetimes may be truncated (and thus knowledge of type (ii) is not needed). Such simpler estimators abound in the literature, but their performance is generally poor (neglecting (i) leads to data-model mismatch, and (ii), to biased estimators). In this work we ignore these simpler estimators and focus on the MoME and, especially, the MLE.} According to the simulations, the MLE performs better than the MoME. As the intensity-aperture product~$\lambda_0L_0$ increases, the MoME lags behind the MLE, which steadily approaches the CRLB and becomes efficient.\footnote{We have numerically verified that both the MLE and the MoME are asymptotically unbiased as $\lambda_0L_0\to\infty.$} As can be seen, estimating~$\lambda$ is easier than estimating~$L_\op{BS}$, where ``easier'' means here that the RMSE approaches $\sqrt{CRLB}$ more rapidly. For instance, the MLE of~$\lambda$ can be considered efficient as early as $\lambda_0L_0=5$, while that of~$L_\op{BS}$ requires $\lambda_0L_0=15$ or larger to be considered within range of the CLRB. A somewhat surprising observation is that the performance of the MLE improves as the ratio $\frac{L_\op{BS}}{L_0}$ increases. (By contrast, the performance of the MoME appears to be rather insensitive to the value of~$L_\op{BS}.$) This can be explained by the fact that as~$L_\op{BS}$ grows larger for fixed~$\lambda_0L_0$, more BS-VR lifetimes are (partially) observed by the PLA, contributing to improve estimation accuracy. Another interesting remark is that the amount of fully observed BS-VR lifetimes is only a fraction of the total; thus, developing estimators that can reliably cope with truncated observations appears to be of great importance.

Applying Theorems~\ref{theo_mle} and~\ref{theo_crlb} to the data from the LOS sites, we obtain the estimates
\begin{align}
  \hat{\lambda} = 2.6\pm 12\%,\quad \hat{L}_\op{BS} = 2.9\pm 13\%.
\end{align}
Compared to Fig.~\ref{fig_mle_std}, the relative RMSEs have been divided by $\sqrt{4}.$ This is because each of the four measurement sites is assumed to produce uncorrelated observations. In view of the above, when designing measurement campaigns with PLAs one should perhaps ensure
\begin{align}
  \lambda_0L_0>15,\quad\text{and}\quad\lambda_0L_0N_0>100
\end{align}
or so, where $N_0$ is here the number of \emph{uncorrelated} experiments. In the next section we discuss when two experiments can be considered uncorrelated.

\subsection{Autocorrelation Function of Observed BS-VRs}Of considerable importance to characterize an stochastic process is the autocorrelation function (ACF) of the process.
\begin{theorem}[Autocorrelation function]\label{theo_cov}Let~$N(x,\vec{r}),$ with $(x,\vec{r})\in \R^+\times\R^2,$ denote the number of far clusters visible from BS array location~$x$ and MS location~$\vec{r},$ and define the covariance function
  \begin{align*}
    C(\Delta x,\Delta\vec{r}) = \E\left\{(N(x+\Delta x,\vec{r}+\Delta\vec{r})-m)(N(x,\vec{r})-m)\right\},
  \end{align*}
  where $m=\E(N(x,\vec{r}))$ is assumed constant. Then, the ACF~$R_N(\Delta x,\Delta\vec{r}) = C(\Delta x,\Delta\vec{r})/C(\vec{0})$ of the number of far clusters in the COST \nolinebreak[1]2100 model is given by
  \begin{align}\label{eq_acf}
    R_N(\Delta x,\Delta\vec{r}) = R_N(\Delta x)\cdot R_N(\Delta\vec{r}),
  \end{align}
  where~$R_N(\Delta x) = \euler^{-\frac{|\Delta x|}{L_\op{BS}}}$ is the BS-side ACF, and
  \begin{align*}
    &R_N(\Delta\vec{r})=
    \begin{cases}
      \frac{1}{\pi}(2\chi - \sin(2\chi)),\quad&\text{for }0\leq|\Delta\vec{r}|\leq 2R_\op{C};\\
      0,\quad&\text{otherwise}
    \end{cases}
  \end{align*}
  is the MS-side ACF, where $\chi=\cos\rmi(\frac{|\Delta\vec{r}|/2}{R_\op{C}})$, and~$R_\op{C}\ge 0$ is the radius of the MS-VRs.
\end{theorem}
\begin{proof}The factorization of the ACF into BS- and MS-side ACFs comes from a separability assumption. Let us first consider the BS-side ACF
  \begin{align}
    R_N(\Delta x) =\frac{\E[(N(x+\Delta x)-m)(N(x)-m)]}{\E[(N(x)-m)^2]}.
  \end{align}
  Since~$R_N(-\Delta x)=R_N(\Delta x)$, we need only to consider the case~$\Delta x>0$. Using the notation introduced in Sec.~\ref{ml_est} with~$x_1 = x$, and~$x_2=x + \Delta x$, we can write$$N(x) = N_{10} + N_{11},\quad N(x+\Delta x) = N_{01} + N_{11}.$$ A moment's reflection will convince the reader that $R_N(\Delta x)$ can only depend upon the number~$N_{11}$ of BS-VRs visible from both~$x$ and~$x+\Delta x$, that is
  \begin{align}
    R_N(\Delta x) = \frac{\Var(N_{11})}{\Var(N(0))}.
  \end{align}
  Using that~$N(0)\in\Po(\lambda\cdot\E(Y))$, and~$N_{11}\in\Po(\lambda\cdot\int_{\Delta x}^\infty(t-\Delta x)f_Y(t)\,dt)$, and after some manipulations, we obtain $R_N(\Delta x) = 1 - \frac{\int_0^\infty\min(t,\Delta x)f_Y(t)\,dt}{\E(Y)},$ which for~$Y\in\Exp(L_\op{BS})$ becomes~$R_N(\Delta x)=\euler^{-\frac{\Delta x}{L_\op{BS}}}$. The computation of the MS-side factor~$R_N(\Delta\vec{r})$ proceeds along the same lines and is omitted.
\end{proof}

If we consider two measurements at locations $(x,\vec{r})$, $(x+\Delta x,\Delta\vec{r})$ to be uncorrelated whenever $R_N(\Delta x,\Delta\vec{r})\leq{1}/{\euler},$ then experiments become uncorrelated if conducted at $|\Delta x|>L_\op{BS}$ apart, or $|\Delta\vec{r}|>(1-{1}/{\euler})2R_\op{C}$ apart, or a combination of both. We should also note the important fact that the ACF of the large-scale fading (LSF),~$R_\op{F}(\Delta x,\Delta\vec{r})$ is also given by Theorem~\ref{theo_cov}. The reason is, intuitively, that if the clusters' powers are assumed independent and identically distributed (which is a reasonable assumption for $|\Delta x|,$ $|\Delta\vec{r}|$ in the order of a few tens of wavelengths) and independent of the VRs, then only the proportion of common clusters affects the ACF.

\subsection{Results and Validation}To gain more insight into the behavior of the proposed BS-VR extension of the model and the influence of the mean lifetime and intensity, simulation results of a PLA and $K=9$ single-antenna users in NLOS propagation conditions are presented in Fig.~\ref{fig_cond}. Users are located in an outdoor environment at distance $R$ from the BS; the angular separation between adjacent MSs is one degree. At the BS side, a ULA with $M=128$ $\lambda/2$-spaced antennas spanning $L=7.5$ meters is used. All antennas are vertically-polarized and omnidirectional in the azimuth plane. Simulation results have been obtained from 100 runs using the parameters in Table~\ref{tab_COST2100}. The carrier frequency is 2.6~GHz and the bandwidth, sampled at 257 equispaced points, 50~MHz.

\begin{figure*}[t]
  \psfrag{condHHdB}[][]{\footnotesize{$\kappa_\op{dB}$\,/\,dB}}
  \psfrag{CDF}[][]{\footnotesize{$\Prob(\kappa_\op{dB}\leq\text{abscissa})$}}
  \psfrag{R60lambda2p3}[][]{\footnotesize{$(R=60, \lambda=2.9)$}}
  \psfrag{lambda2p3Leff3}[][]{\footnotesize{$(\lambda=2.9, L_\op{BS}=3.2)$}}
  \psfrag{R60Leff3}[][]{\footnotesize{$(R=60, L_\op{BS}=3.2)$}}
  \psfrag{R60lambda2p3}[][]{\footnotesize{$(R=60, \lambda=2.6)$}}
  \psfrag{lambda2p3Leff3}[][]{\footnotesize{$(\lambda=2.6, L_\op{BS}=3)$}}
  \psfrag{R60Leff3}[][]{\footnotesize{$(R=60, L_\op{BS}=3)$}}
  \psfrag{Leffeq1mxxxx}[][]{\scriptsize{$L_\op{BS}\!=\!1.1$}}
  \psfrag{Leffeq3mxxxx}[][]{\scriptsize{$L_\op{BS}\!=\!3.2$}}
  \psfrag{Leffeq6mxxxx}[][]{\scriptsize{$L_\op{BS}\!=\!6.4$}}
  \psfrag{Leffeq12mxxxx}[][]{\scriptsize{$L_\op{BS}\!=\!12.8$}}
  \psfrag{Leffeq1m}[][]{\scriptsize{$L_\op{BS}\!=\!1$}}
  \psfrag{Leffeq3m}[][]{\scriptsize{$L_\op{BS}\!=\!3$}}
  \psfrag{Leffeq6m}[][]{\scriptsize{$L_\op{BS}\!=\!6$}}
  \psfrag{Req10m}[][]{\scriptsize{$R\!=\!10$}}
  \psfrag{Req30m}[][]{\scriptsize{$R\!=\!30$}}
  \psfrag{Req60m}[][]{\scriptsize{$R\!=\!60$}}
  \psfrag{Req100m}[][]{\scriptsize{$R\!=\!100$}}
  \psfrag{Req10mxxxx}[][]{\scriptsize{$R\!=\!10$}}
  \psfrag{Req30mxxxx}[][]{\scriptsize{$R\!=\!30$}}
  \psfrag{Req60mxxxx}[][]{\scriptsize{$R\!=\!60$}}
  \psfrag{Req100mxxxx}[][]{\scriptsize{$R\!=\!100$}}
  \psfrag{laeq1p15}[][]{\scriptsize{$\lambda\!=\!1.15$}}
  \psfrag{laeq2p30}[][]{\scriptsize{$\lambda\!=\!2.60$}}
  \psfrag{laeq4p60}[][]{\scriptsize{$\lambda\!=\!5.20$}}
  \psfrag{laeq9p20}[][]{\scriptsize{$\lambda\!=\!10.40$}}
  \psfrag{laeq1p15xxxx}[][]{\scriptsize{$\lambda\!=\!1.45$}}
  \psfrag{laeq2p30xxxx}[][]{\scriptsize{$\lambda\!=\!2.90$}}
  \psfrag{laeq4p60xxxx}[][]{\scriptsize{$\lambda\!=\!5.80$}}
  \psfrag{laeq9p20xxxx}[][]{\scriptsize{$\lambda\!=\!11.60$}}
  \centering
  \includegraphics[width=\textwidth]{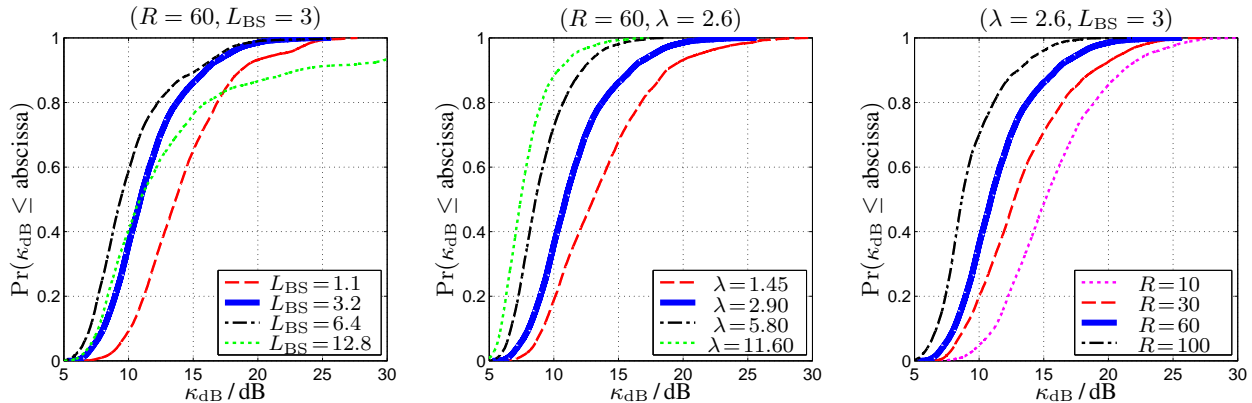}
  \caption{Comparison of the channel condition number,~$\kappa_\op{dB},$ of $K=9$ single-antenna users in an outdoor environment with NLOS propagation conditions to a PLA with $M=128$ antennas. (Left) Varying~$L_\op{BS}$ for fixed~$R, \lambda;$ (middle) Varying~$\lambda$ for fixed~$R, L_\op{BS};$ (right) Varying~$R$ for fixed~$L_\op{BS}, \lambda.$}
  \label{fig_cond}
\end{figure*}
We study the condition number, $\kappa(\vec{H})=\frac{\sigma_1(\vec{H})}{\sigma_K(\vec{H})},$ of the $K\times M$ MU-MIMO channel matrix~$\vec{H}.$ In particular, CDFs of
\begin{align}
  \kappa_\op{dB}(\vec{H}(t,f)) = 20\log_{10}(\kappa(\vec{H}(t,f)))  
\end{align}
are shown, of which several interesting remarks can be made. Firstly,~$\kappa_\op{dB}$ can decrease as the mean BS-VR lifetime,~$L_\op{BS},$ increases (for fixed $\lambda$), which is in contrast to common wisdom stipulating that spatial separability improves as more array sections are exposed to different environments. Secondly,~$\kappa_\op{dB}$ decreases as the BS-VR process birth-rate,~$\lambda,$ increases for fixed $L_\op{BS}>0.$ In this case, the reason is that larger values of $\lambda$ entail richer environments with more scatterers (clusters), which facilitates spatial separability. Thirdly,~$\kappa_\op{dB}$ decreases as the BS-MS distance,~$R,$ increases for fixed MS angular separation. By the discussion in the preceding section, the MS-side LSF ACF,~$R_N(\Delta\vec{r})$, is low for~$\frac{R}{R_\op{C}}>>1,$ thereby driving toward zero the condition number,~$\kappa_\op{dB}.$ Conversely,~$R_N(\Delta\vec{r})$ goes to one as $\frac{R}{R_\op{C}}\to 0$, the overall effect being a noticeable increase of~$\kappa_\op{dB}$. (This is despite the fact that spherical effects become more noticeable at close range, pulling the BS-side small-scale fading (SSF) ACF,~$R_N(x_1,x_2),$ toward zero.)

\section{Extension for Closely-Located Users}\label{sec_extension_2}

Using indoor measured channels, we next characterize, analyze, and validate an extension of the COST \nolinebreak[1]2100 model that describes individual MPCs at the MS side in terms of birth-death processes, and which is of particular relevance for MaMi systems with closely-located users.

\subsection{Distribution and Intensity of MPC Lifetimes}

The results of this section are based on channel measurements in a large sports hall (20~m\allowbreak$\times$36~m\allowbreak$\times$7.5~m) in the 2.6~GHz band. The radio channel between a conical monopole omnidirectional antenna moving on a 15-m-long straight-line track and a fixed cylindrical array with 128 antennas was sampled at an average rate of 400 snapshots/m. The snapshots obtained in this way were then processed using a phase-based extended Kalman filter (EKF) that was able to track the MPC parameters MPC of the propagation channel along a 15~m route. The reader is referred to~\cite{Li:2018:ekf} for further discussion of the measurements and their processing.

To analyze the lifetimes of MPCs, we adopt the modeling framework introduced in Sec.~\ref{sec_model}. Let $\mathcal{I}=\{1,\ldots,n\},$ where $n$ is the number of detected MPCs. The output of the EKF consists then of a sequence~$a_1\leq\ldots\leq a_n$ of \emph{observed} MPC birth locations such that~$x_1\leq a_i\leq x_2$ for all $i\in\mathcal{I}$, and a sequence~$\upsilon_1,\ldots,\upsilon_n$ of \emph{observed} MPC lifetimes\footnote{Technically, the observed MPC lifetimes~$\upsilon_1,\ldots,\upsilon_n$ are outcomes of the random variables $\Upsilon_1,\ldots,\Upsilon_n.$ These are assumed identically distributed, and are denoted generically by~$\Upsilon.$} such that~$\Delta_0\leq\upsilon_i\leq L$ for all $i\in\mathcal{I}$ and for some minimum feature size~$\Delta_0>0.$ For the particular track that we will study we have~$L=15$~m,~$\Delta_0=7.5$~cm, and~$n=752.$ 

As before, let~$\vec{\theta} = [\lambda,\vec{\theta}_Y\rmt]\rmt$ be a $(p+1)$-dimensional vector of parameters to be estimated, and let $\vec{x}=[a_1,b_1,a_2,\dots,b_n]$ be a $2n$-dimensional vector of data. To model the MPC lifetimes we restrict attention to the following three cases:
\begin{enumerate}
\item[1)]$Y\in\Exp(L_\op{MPC})$. Thus, $f_Y(y;\vec{\theta}_Y)$ takes the form~\eqref{eq_exp_pdf} with~$\vec{\theta}_Y=[L_\op{MPC}]$ for some~$L_\op{MPC}\in\R^+$, and the likelihood function~$\mathcal{L}(\vec{\theta};\vec{x})$ is given by~\eqref{eq_like_exp}. As this and the BS-VR problem addressed in Sec.~\ref{ml_est} are formally identical, it follows that the vector~$\hat{\vec{\theta}} = [\hat{\lambda},\hat{L}_\op{MPC}]\rmt$ of MLEs is available in closed-form from~\eqref{eq_lambda_hat} and~\eqref{eq_lbs_hat} in Theorem~\ref{theo_mle}.
\item[2)]$Y\in\Lognormal_{10}(\mu_\op{MPC},\sigma^2_\op{MPC})$ is lognormally distributed. Thus, $f_Y(y;\vec{\theta}_Y)$ takes the form
  \begin{align}\label{eq_ln10}
    f_Y(y;\vec{\theta}_Y) =
    \begin{cases}
      0,\quad&\text{for }y<0;\\
      \frac{1}{y\sqrt{\psi 2\pi}}\exp(-\frac{(\log(y)-m)^2}{2\psi}),\quad&\text{for }y\ge 0,
    \end{cases}
  \end{align}
  with $\vec{\theta}_Y=[\mu_\op{MPC},\sigma^2_\op{MPC}]\rmt$ for some $\mu_\op{MPC}\in\R,$ $\sigma^2_\op{MPC}\in\R^+$ such that$$m = \mu_\op{MPC}(\frac{\log(10)}{10}),\quad\psi = \sigma^2_\op{MPC}(\frac{\log(10)}{10})^2,$$and~$\mathcal{L}(\vec{\theta};\vec{x})$ is given by~\eqref{eq_like}. As no closed-form solution is available in this case, the vector~$\hat{\vec{\theta}} = [\hat{\lambda},\hat{\mu}_\op{MPC},\hat{\sigma}^2_\op{MPC}]\rmt$ of MLEs has been obtained by solving~\eqref{eq_mle} numerically.\footnote{To facilitate the numerical evaluations in~\eqref{eq_mle}, closed-form expressions, as a function of~$\Phi(x),$ of some of the terms involved are given in~\cite{Flordelis19}.}
\item[3)]$\Upsilon\in\TLN_{10}(\mu_\op{MPC},\sigma^2_\op{MPC},\Delta_0,L),$ the \emph{observed} MPC lifetimes, follows a truncated lognormal distribution. In this case $f_\Upsilon(\upsilon;\vec{\theta}_\Upsilon)$ is given by
  \begin{align}\label{eq_tln10}
    f_\Upsilon(\upsilon;\vec{\theta}_\Upsilon,\Delta_0,L) \propto f_Y(\upsilon;\vec{\theta}_\Upsilon)1_{[\Delta_0,L]}(\upsilon)
  \end{align}
  with~$\vec{\theta}_\Upsilon\!=\![\mu_\op{MPC},\!\sigma^2_\op{MPC}]\rmt$, and~$\mu_\op{MPC},\sigma^2_\op{MPC},f_Y$ as in case 2. The vector~$\hat{\vec{\theta}}_\Upsilon\!=\![\hat{\mu}_\op{MPC},\!\hat{\sigma}^2_\op{MPC}]\rmt$ has been selected to minimize the vertical distance between the empirical CDF, obtained from channel measurements, and~$F_\Upsilon(\upsilon;\vec{\theta}_\Upsilon)$, the reference CDF, with~$\vec{\theta}_\Upsilon$ in some range of interest;~$\hat{\lambda}$ has simply been set to $\hat{\lambda} = n/L.$ We note that this case seems to be the most common approach in the literature~\cite{Mahler:2016:lifetime,Li:2018:ekf}.
\end{enumerate}

\begin{figure}[!t]
  \psfrag{closeup}[][]{\footnotesize{Close-up}}
  \psfrag{measurement}[][]{\scriptsize{Measurements}}
  \psfrag{exponentialfit}[][]{\scriptsize{Exponential}}
  \psfrag{lognormalfit}[][]{\scriptsize{Lognormal}}
  \psfrag{alternativefit}[][]{\scriptsize{Trunc. logn.}}
  \psfrag{upsilon}[][]{\footnotesize{$\upsilon$/m}}
  \psfrag{probupsilon}[][]{\footnotesize{$\Prob(\upsilon)\leq\text{abscissa}$}}
  \centering
  %trim={<left> <lower> <right> <upper>}
  \includegraphics[trim={0 0 0 .2cm},clip,width=.65\textwidth]{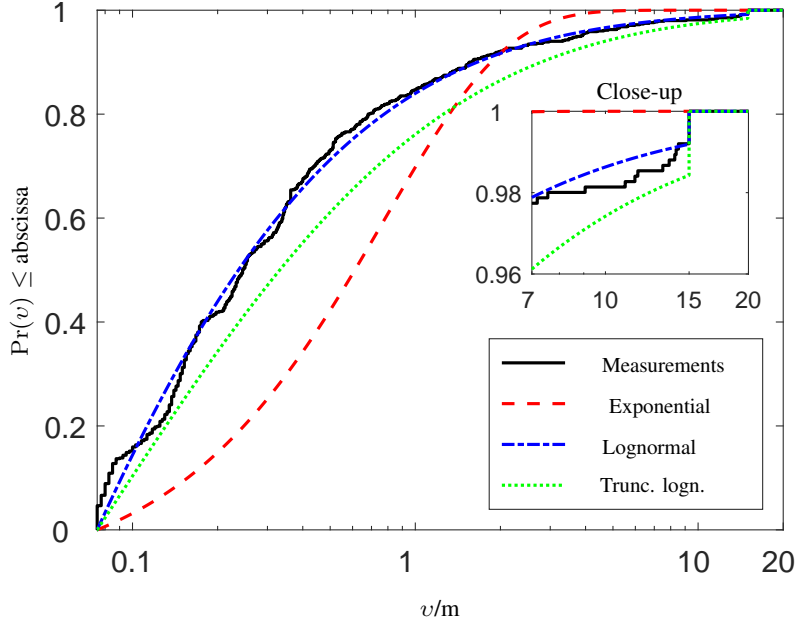}
  \caption{CDF of observed MPC lifetimes~$\upsilon_1,\ldots,\upsilon_n$ in a measured indoor scenario, and corresponding exponential, lognormal, and truncated lognormal fits ($L=15$~m, $\Delta_0=7.5$~cm, $n=752.$)}
  \label{fig_mpc_lifetimes}
\end{figure}
Fig.~\ref{fig_mpc_lifetimes} shows a comparison of the ECDF obtained from measurements, and the CDFs corresponding to cases 1), 2), and 3). According to the figure, the best fit is the lognormal one (case 2). An exponential fit (case 1) matches measurements poorly, underestimating the cumulative density function for short- and medium-lived MPCs ($\upsilon\leq 2$~m) and overestimating it for long-lived ones ($\upsilon>2$~m). The lognormal fit of case 2, on the other hand, shows an excellent agreement over the entire range of the observations. Even the proportion of lifetimes exceeding~$L$, namely~$\Prob(Y>L)$, is best explained by the lognormal distribution, as shown in the inset. Turning to the truncated lognormal case, the figure shows that when observed ($\Upsilon$) rather than true ($Y$) MPC lifetimes are modeled, the resulting fit may not be satisfactory even when the ``correct'' (lognormal) distribution is being used (see also~\cite{Mahler:2016:lifetime,Li:2018:ekf}). The problem is that the observation model of case 3) takes observations~$\upsilon_1,\ldots,\upsilon_n$ at face value, and neglects the fact that MPCs may extend beyond the measurement track, at one or both ends. The data shown in Table~\ref{tab_mpc_lifetimes} further emphasizes this point.
\begin{table}[!b]
  \renewcommand{\arraystretch}{1.5}
  \scriptsize
  \centering
  \caption{Estimated parameters and means of cases 1, 2, and 3.}
  \label{tab_mpc_lifetimes}
  \begin{tabular}{c||ccc}
    \hline
    Case&$\hat{\lambda}$&$\hat{\vec{\theta}}_Y$&$\E(Y;\hat{\vec{\theta}}_Y)$\\
    \hline\hline
    1&52.30&$\hat{L}_\op{MPC}=0.81$&0.81\\
    2&171.60&$\hat{\mu}_\op{MPC}=-16.92,$\quad$\hat{\sigma}^2_\op{MPC}=94.60$&0.25\\
    3&52.13&$\hat{\mu}_\op{MPC}=-11.09,$\quad$\hat{\sigma}^2_\op{MPC}=89.91$&0.84\\
    \hline
  \end{tabular}
\end{table}

\subsection{Distribution of MPC-VR Radii}\label{dist_radii}

In order to model the lifetime of individual MPCs within a cluster we introduce the concept of MPC visibility region (MPC-VR), which is schematically illustrated in Fig.~\ref{fig_MPCVR} in Sec.~\ref{sec_COST2100_intro}. In the figure, MPC-VRs are represented by discs contained within a certain MS-VR. There is a one-to-one correspondence between the MPC-VRs of a certain MS-VR, and the scattering points of the associated cluster. When the MS enters a MPC-VR, the corresponding scattering point becomes active and contributes to the propagation channel. To give an example, active MPC-VRs in Fig.~\ref{fig_MPCVR} have been marked gray. It is a simple but crucial observation that different MSs do not necessarily ``see'' the same set of active MPCs, even when they do have the same set of active MS-VRs; we will apply this later to establish that the SSF ACF decays more rapidly than when the concept of MPC-VR is not applied.

Let a MS move along the $X$-axis, and let MPC-VRs be distributed uniformly on the $XY$-plane. Then, the distribution of the MPC lifetimes encountered during the movement (i.e., the length of the intersecting chords) is
\begin{align}\label{eq_circ}
  \breve{F}_Y(y;R_\op{MPC}) =
  \begin{cases}
    0,\quad&\text{ for }-\infty\leq y< 0;\\
    1-\sqrt{1-(\frac{y}{2R_\op{MPC}})^2},\quad&\text{ for }0\leq y< 2R_\op{MPC};\\
    1,\quad&\text{ for }2R_\op{MPC}\leq y\leq\infty,
  \end{cases}
\end{align}
where~$R_\op{MPC}\in\R^+$ is the radius of the MPC-VRs. (For simplicity, we have assumed that the same set of MS-VRs is always visible during the movement. This assumption is often made in indoor environments~\cite{Zhu:2014:cluster,Cardona:2016:COST}.) The CDF~\eqref{eq_circ} has been plotted in Fig.~\ref{fig_circ}, top, for several values of~$R_\op{MPC}.$

To reconcile~\eqref{eq_circ} with our previous findings, namely that true MPC lifetimes appear to be lognormally distributed, we need to find~$R_\op{MPC}\in\R^+$ such that~$\breve{F}_Y(y;R_\op{MPC})$ approaches
\begin{align}\label{eq_Phi}
  F_Y(y;\hat{\vec{\theta}}_Y) = \Phi(\frac{\log(y)-m}{\psi^{1/2}})
\end{align}
with~$\hat{\vec{\theta}}_Y$ given by case 2 in Table~\ref{tab_mpc_lifetimes}, and~$m,\psi$ as before. However, a quick examination of Fig.~\ref{fig_circ}, top, shows that no such number exists: The desired lognormal curve~\eqref{eq_Phi} is concave for $y\ge 0$, but for every fixed~$R_\op{MPC}\in\R^+$,~\eqref{eq_circ} turns out to be convex in the interval $0\leq y\leq 2R_\op{MPC}.$

The solution to this difficulty is to use discs of various sizes, an idea which is supported by the measurements. Let~$R_\op{MPC}$ be a random variable with PDF~$f_{R_\op{MPC}}:\R^+\to\R^+$. Then by (the continuous version of) the law of total probability~\cite{Gut:2009:P} we have
\begin{align}\label{eq_circ_tot}
  \breve{F}_Y(y) = \frac{1}{\E(R_\op{MPC})} \int_{-\infty}^\infty r\,\breve{F}_Y(y;r)\,f_{R_\op{MPC}}(r)\,dr.
\end{align}
Thus, $\breve{F}_Y(y)$ can be thought of as a ``weighted sum'' of MPC-VR functions~\eqref{eq_circ}, and our task is to find nonnegative weights~$f_{R_\op{MPC}}(r)$ for all $r\in\R^+$ such that~$\breve{F}_Y(y)$ is close to lognormal according to some prescribed criterion of closeness. Our approach here is to adopt the least-squares criterion. As we have no a priori reason to expect~$f_{R_\op{MPC}}(r)$ to be easily obtainable by analytical methods, we solve a discretized version of problem~\eqref{eq_circ_tot}, namely
\begin{align}\label{eq_discrete}
  \minimize_{\vec{w}}\quad&(\vec{A}\vec{w}-\vec{b})\rmh(\vec{A}\vec{w}-\vec{b})\nonumber\\
  \text{subject to}\quad&w_1,\ldots,w_q\geq 0,\quad\sum_{i=1}^qw_i=1,
\end{align}
where $\vec{A}$ is a~$p\times q$ matrix with $[\vec{A}]_{ij}=\breve{F}_Y(y_i;r_j)$ as the $ij$-th entry, $\vec{w}$ is a $q\times 1$ vector of nonnegative weights~$w_1,\ldots,w_q$, and $\vec{b}$ is a $p\times 1$ vector with $F_Y(y_i;\vec{\theta})$ as the $i$-th entry. Further, we have defined the sets $\mathcal{Y}=\{y_1,\ldots,y_p\}$ of sampling lifetimes, and $\mathcal{R}=\{r_1,\dots, r_q\}$ of sampling radii. Note that the condition~$\sum_{i=1}^qw_i=1$ is needed to ensure that the mapping
\begin{align}\label{eq_discrete_w}
  p_{R_\op{MPC}}(r_i) = \Prob(R_\op{MPC}=r_i) = w_i,\quad r_i\in\mathcal{R},
\end{align}
is a proper PMF. Crucially, the optimization problem~\eqref{eq_discrete} is a quadratic program~\cite{boyd04}, and can be readily solved by using a number of available numerical packages, such as MATLAB or CVX~\cite{cvx}.\footnote{A different version of~\eqref{eq_discrete} sets $[\vec{A}]_{ij}=r_i\,\breve{F}_Y(y_i;r_j),$ $1\leq i\leq p,$ $1\leq j\leq q,$ and uses the mapping $p_{R_\op{MPC}}(r_i) = (\frac{w_i}{r_i})/(\sum\limits_{k=1}^p\frac{w_k}{r_k})$ instead of~\eqref{eq_discrete_w}. While closer to the original, continuous formulation, this alternative way of discretizing~\eqref{eq_circ_tot} is prone to numerical issues and instability for small values of $r_i.$ Because of this, in this work we restrict attention to~\eqref{eq_discrete},~\eqref{eq_discrete_w}.}

\begin{figure}[!t]
  \psfrag{Req12ldots23}[Bl][Bc]{\footnotesize{$R_\op{MPC}=1,2,\ldots,23$}}
  \psfrag{FYcircle}[][]{\footnotesize{$\breve{F}_Y(y;R_\op{MPC})$}}
  \psfrag{y}[Bc][tc]{\footnotesize{y}}
  \psfrag{FYtarget}[][]{\footnotesize{$F_Y(y;\vec{\theta})$}}
  \psfrag{r}[Bc][tc]{\footnotesize{r}}
  \psfrag{FRtarget}[][]{\footnotesize{$F_{R_\op{MPC}}(r)$}}
  \psfrag{yA}[][]{\scriptsize{\eqref{eq_ln10}}}
  \psfrag{yB}[][]{\scriptsize{\eqref{eq_discrete}}}
  \psfrag{yC}[][]{\scriptsize{\eqref{eq_discrete}}}
  \psfrag{yD}[][]{\scriptsize{\eqref{eq_discrete_ln10}}}
  \centering
  %trim={<left> <lower> <right> <upper>}
  \includegraphics[trim={0 0 0 .2cm},clip,width=.60\textwidth]{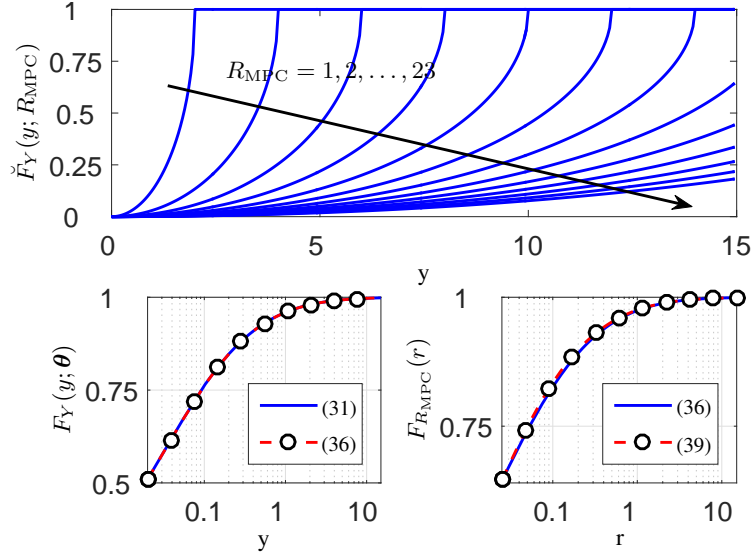}
  \caption{Empirical CDF of the observed lifetimes,~$\{\upsilon_i\}$, of the MPCs in a measured indoor scenario, and CDFs of exponential, lognormal, and truncated lognormal fits ($\Delta_0=0.075$~m, $L=15$~m). See also Table~\ref{tab_mpc_lifetimes}.}
  \label{fig_circ}
\end{figure}
Fig.~\ref{fig_circ}, bottom left, shows a comparison of the lognormal CDF~\eqref{eq_Phi} and its approximation $\tilde{F}_Y(y_i)$ obtained by solving~\eqref{eq_discrete} with~$\mathcal{Y}=\{.0025,\allowbreak.0525,\allowbreak\ldots,\allowbreak 14.9525\}$ and~$\mathcal{R}=\{.000,\allowbreak.025,\allowbreak\ldots,\allowbreak 23.000\}$, so that~$\vec{A}$ is of size $300\times 921.$ The approximation is excellent, with a root-mean-square error
\begin{align}
  \RMSE = \sqrt{\frac{\sum_{y_i\in\mathcal{Y}}(\tilde{F}_Y(y_i)-F_Y(y_i;\vec{\theta}))^2}{\sum_{y_i\in\mathcal{Y}}1}}
\end{align}
about $1.1\times 10^{-14}$ in the range $0\leq y\leq L.$ The goodness of the fit suggests that the proposed method, namely to model MPC birth-death processes using MPC-VRs with radii~$R_\op{MPC}$ drawn from a certain distribution $f_{R_\op{MPC}}:\R^+\to\R^+$, is rather flexible and, as such, applicable to a greater range of environments having different propagation characteristics. For each such environment, the appropriate distribution of~$R_\op{MPC}$ can be found by solving an optimization problem like~\eqref{eq_discrete}.

In our case, however, a further generalization is possible: Fig.~\ref{fig_circ}, bottom right, shows the cumulative sum of the weights~$w_1,\ldots,w_{921}$, and a lognormal approximating curve which has the form~\eqref{eq_Phi} with parameters
\begin{align}\label{eq_discrete_ln10}
  \mu_{R_\op{MPC}}=-19.8,\qquad\sigma^2_{R_\op{MPC}}=101.3.
\end{align}
The radii~$R_\op{MPC}$ are now drawn from a two-parameter continuous distribution which conveniently summarizes all $q=921$ weights. Again, the agreement between the two curves is excellent, with a $\RMSE\approx2.2\times 10^{-4}$ whenever~$R_\op{MPC}>.025.$ Motivated by this, we introduce in the COST \nolinebreak[1]2100 model circular MPC-VRs with lognormally distributed radii according to the parameters~$\mu_{R_\op{MPC}},$~$\sigma^2_{R_\op{MPC}}$, given in Table~\ref{tab_COST2100}.

\subsection{Distribution of MPC-VR Amplitudes}\label{sec_gain_func}We model the relative contribution of each MPC to the total channel gain by means of the so-called \emph{gain function}~\cite{Gao:2015:MAMI,Gao:2015:Extension}. More precisely, a Gaussian profile
\begin{align}\label{eq_weighting}
  g_{\op{MPC},\ell}(\vec{r}_\op{MS}) = \exp(-\frac{\dist(\vec{r}_\op{MS}-\vec{r}_{\op{g},\ell})^2}{2\sigma^2_{\op{g},\ell}})    
\end{align}
multiplies the complex amplitude of each MPC, where the weight~$g_{\op{MPC},\ell}(\vec{r}_\op{MS})$ depends on the Euclidean distance between the user position~$\vec{r}_\op{MS}=[r_{\op{MS},x},\allowbreak r_{\op{MS},y},\allowbreak r_{\op{MS},z}]\rmt$ and the center~$\vec{r}_{\op{g},\ell}=[r_{\op{g},\ell,x},r_{\op{g},\ell,y}]\rmt$ of the $\ell$-th MPC-VR measured on the $XY$-plane:
\begin{align*}
  \dist(\vec{r}_\op{MS}-\vec{r}_{\op{g},\ell})=\sqrt{(r_{\op{MS},x}-r_{\op{g},\ell,x})^2 + (r_{\op{MS},y}-r_{\op{g},\ell,y})^2}.
\end{align*}
Thus, at distance, e..g., $\sigma_{\op{g},\ell}$ from its center, the MPC gain function has decayed by 4.3~dB. The centers~$\vec{r}_{\op{g},\ell}$ are generated uniformly inside the corresponding MS-VRs, which are assumed to have radius~$R_\op{C}\in\R^+$. The profile widths~$\sigma_{\op{g},\ell}$ are identified with the radii~$R_\op{MPC}$, and follow a lognormal distribution with parameters see~\eqref{eq_discrete_ln10}. Thanks to the MPC gain function feature, variations observed~\cite{Jost:12:KEST,Wang:2012:lifetime,Li:2018:ekf} in the gain of individually tracked MPCs can be adequately modeled, helping to smoothly ramp up (fade out) the contributions of new-born (dead) MPCs, as users navigate in an out of the MPC-VRs. 

Let~$N_\op{MPC}^\op{eff}$ denote the average effective number of MPCs per cluster, defined as the average number of MPCs whos amplitude at any given point within the MPC-VR exceeds a certain value. Then, the required number of MPCs per cluster is
\begin{align}
  N_\op{MPC} = N_\op{MPC}^\op{eff}\frac{R_\op{C}^2}{\E(R_\op{MPC}^2)} = N_\op{MPC}^\op{eff}\left(\frac{R_\op{C}}{\exp(m'+\psi')}\right)^2,
\end{align}
where
\begin{align}
  m' = \mu_{R_\op{MPC}}(\frac{\log(10)}{10}),\quad\psi'=\sigma_{R_\op{MPC}}^2=(\frac{\log(10)}{10})^2,
\end{align}
and $\log$ denotes the natural logarithm. We remark that if backwards compatibility is desired, one can simply set $N_\op{MPC} = N_\op{MPC}^\op{eff},$ $\mu_{R_\op{MPC}}=\infty,$ $\sigma_{R_\op{MPC}}=0,$ and the MPC-VRs concentric with the MS-VRs.

\subsection{Results and Validation}
To validate the proposed extension, we here study the autocorrelation function and the channel condition number.

\subsubsection{Autocorrelation Function}\label{sec_acf_mpc}

One important result in this paper is the following:

\begin{theorem}[ACF of SSF]\label{theo_acf_ssf}The space-frequency ACF of the SSF is asymptotically given by
  \begin{align}
    R(\Delta\vec{r},\Delta f)\to R_\op{H}(\Delta\vec{r},\Delta f)R_\op{Y}(\Delta\vec{r})\text{ as }\E(n)\to\infty,
  \end{align}
  where $R_\op{H}(\Delta\vec{r},\Delta f)$ is the ACF of the SSF \emph{without} the gain function, and
  \begin{align}
    R_\op{Y}(\Delta\vec{r}) = \frac{\E(n_{11}(\Delta\vec{r}))}{\E(n)}.  
  \end{align}
  In particular,~$R(\Delta\vec{r},\Delta f)$ obeys the bound
  \begin{align}
    |R(\Delta\vec{r},\Delta f)| \leq |R_Y(\Delta\vec{r})|.
  \end{align}
\end{theorem}
\begin{proof}We start by writing the propagation channel of the COST2100 model as
  \begin{align}\label{eq_cost_channel}
    H(\vec{r},f) = \sum_{\ell\in\mathcal{N}(\vec{r})} a_\ell\euler^{-\ramuno(\vec{k}_\ell\!\cdot\!\vec{r} + 2\pi\tau_\ell f)},
  \end{align}
  where~$\vec{k}_\ell=\frac{2\pi}{\lambda_0}\hat{\vec{k}}_\ell$ is the (vector) wavenumber,~$\hat{\vec{k}}_\ell$ is a unit vector pointing in the direction of propagation of the~$\ell$-th MPC, and we assume that the complex amplitudes~$a_\ell$ (and thus~$H(\vec{r},f))$ have zero mean. Note that the subindices~$\ell$ can possibly extend over several MPC-VRs. For any $\vec{r}\in\R^3,f\in\R$, the covariance function of the SSF is given by
  \begin{align}\label{eq_tmp_C}
    &C(\Delta\vec{r},\Delta f)\nonumber\\
    &= \E\left\{H(\vec{r},f)\rmc H(\vec{r}+\Delta\vec{r},\Delta f)\right\}\nonumber\\
    &= \E_{\mathcal{N}(\cdot),a_\ell}\!\!\left\{\!\!\left(\sum_{\ell\in\mathcal{N}(\vec{r})}a_\ell\euler^{-\ramuno(\vec{k}_\ell\!\cdot\!\vec{r} + 2\pi\tau_\ell f)}\right)\rmc \!\!\!\left(\sum_{\ell'\in\mathcal{N}(\vec{r}+\Delta\vec{r})}a_{\ell'}\euler^{-\ramuno(\vec{k}_{\ell'}\!\cdot\!(\vec{r}+\Delta\vec{r}) + 2\pi\tau_{\ell'} f)}\right)\!\!\right\}\!\!.
  \end{align}
  Now, by defining~$\mathcal{N}_{11}:=\mathcal{N}(\vec{r}) \cap \mathcal{N}(\vec{r}+\Delta\vec{r})$, and assuming uncorrelated scattering,~\eqref{eq_tmp_C} reduces to
  \begin{align*}
    C(\Delta\vec{r},\Delta f) &= \E_{\mathcal{N}_{11}}\{\sum_{\ell\in\mathcal{N}_{11}}\E(|a_\ell|^2)\euler^{-\ramuno(\vec{k}_\ell\!\cdot\!\Delta\vec{r} + 2\pi\tau_\ell\Delta f)}\}\nonumber\\
    &\approx R_Y(\Delta\vec{r})\E_{\mathcal{N}(\vec{r})}\{\sum_{\ell\in\mathcal{N}(\vec{r})}\E(|a_\ell|^2)\euler^{-\ramuno(\vec{k}_\ell\!\cdot\!\Delta\vec{r} + 2\pi\tau_\ell\Delta f)}\}\\
    &\to R_Y(\Delta\vec{r})C_H(\vec{r},f)\text{ as }\E(n)\to\infty,
  \end{align*}
  where~$C_H(\vec{r},f)$ is the covariance function of the SSF without the gain function. Finally, normalize~$C(\Delta\vec{r},\Delta f)$ by dividing by $C(0,0)=C_H(0,0)$ and the claim follows.
\end{proof}

\begin{example}\label{ex_1}Let~$R_\op{MPC}$ have distribution~$F_{R_\op{MPC}}(r)=\delta_{R_0}(r)$ for some $R_0\ge 0$, where~$\delta_{x_0}(x)$ is the Dirac delta function. Then the ACF of the number of MPCs has the form
  \begin{align}\label{eq_circ_acf}
    R_\op{Y}(d) =
    \begin{cases}
      \frac{1}{\pi}(2\chi-\sin(2\chi)),\quad&\text{ for }0\leq d\leq 2R_0;\\
      0,\quad&\text{ otherwise},
    \end{cases}
  \end{align}
  where $\chi=\cos\rmi(\frac{d/2}{R_0})$ (and so~$0\leq\chi\leq\frac{\pi}{2}$), and~$d=|\Delta\vec{r}|$ (cf. BS-side ACF of Theorem~\ref{theo_cov}). The autocorrelation function~\eqref{eq_circ_acf} is sometimes called the circular correlation function~\cite{Matern:1960,Abrahamsen:1997}. Note that~\eqref{eq_circ_acf} vanishes beyond $d=R_0$, and so the SSF becomes uncorrelated.
\end{example}

\begin{example}Let~$R_\op{MPC}$ have arbitrary PDF~$f_{R_\op{MPC}}:\R^+\to\R^+$ with finite second-order moment. Then the ACF of the number of MPCs takes the form
  \begin{align}\label{eq_circ_acf_gen}
    R_\op{Y}(d) = \frac{\frac{1}{\pi}\int_{d/2}^\infty r^2(2\chi(r)-\sin(2\chi(r))f_{R_\op{MPC}}(r)\,dr}{\int_0^\infty r^2f_{R_\op{MPC}}(r)\,dr},
  \end{align}
  with~$\chi(r) = \cos\rmi(\frac{d/2}{r})$ and, consequently, $\sin(2\chi(r)) = 2(\frac{d/2}{r})\sqrt{1-(\frac{d/2}{r})^2}.$
\end{example}

The examples above show that the MPC gain function can play an important role at reducing multiuser correlation, especially when the channel as seen by the users exhibits small angular spreads, or when the MPC lifetimes are particularly short (see Example~\ref{ex_1}). For typical angular spreads in outdoor and (especially) indoor environments, however,~$R_\op{H}(\Delta\vec{r},\Delta f)$ often decays rapidly with~$|\Delta\vec{r}|$, in effect rendering the contribution of~$R_Y(\Delta\vec{r})$ less apparent. In fact, a careful comparison\footnote{Due to lack of space, the detailed results of this comparison are omitted, and only a short summary is given. See towards the end of the next section for related comments.} of user correlation levels in (i) channels measured in an indoor environment, (ii) synthetic channels emulating the said environment and including the proposed gain function extension, and (iii) synthetic channels without the gain function extension, revealed no statistically significant differences between the three cases ($|\Delta\vec{r}|$ being in the range between 0.5 and 3.5~m). Notwithstanding the above, the importance of modeling the MPC gain function becomes clear when higher-order statistics, involving more than two users, are considered, as we show next.

\subsubsection{Channel Condition Number}

\begin{figure}[!t]
  \psfrag{dB}[][]{\footnotesize{$\kappa_\op{dB}$\,/\,dB}}
  \psfrag{Probkappaltabscissa}[][]{\footnotesize{$\Prob(\kappa_\op{dB}\leq\text{abscissa})$}}
  \psfrag{title1}[][]{\footnotesize{$M=128$}}
  \psfrag{title2}[][]{\footnotesize{$M=32$}}
  \psfrag{ON}[][]{\scriptsize{ON}}
  \psfrag{OFF}[][]{\scriptsize{OFF}}
  \psfrag{Directive}[][]{\scriptsize{Directive}}
  \psfrag{Omni}[][]{\scriptsize{\makebox[.05cm]{}Omni.}}
  \psfrag{Measured}[][]{\scriptsize{Measured}}
  \centering
  \includegraphics[width=.65\textwidth]{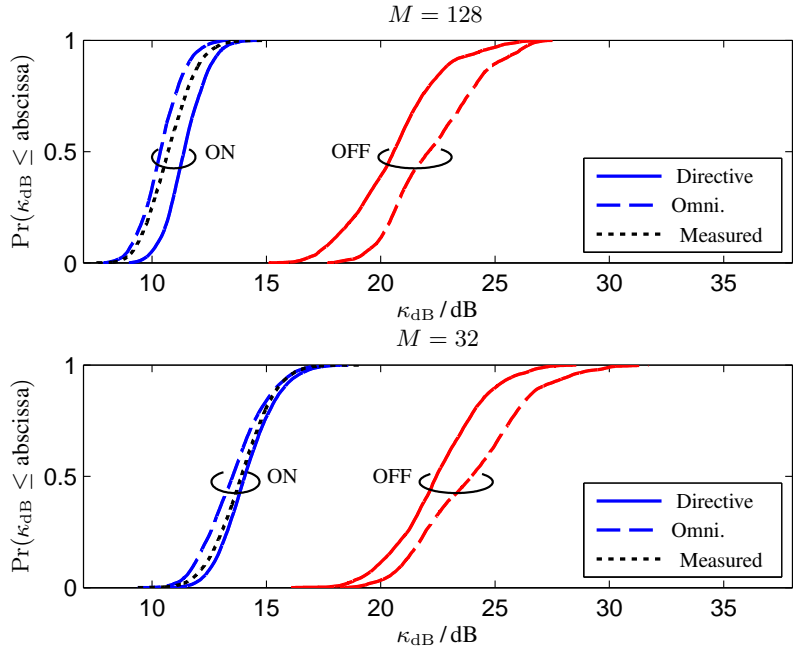}
  \caption{Comparison of simulated MU-MIMO channels of nine closely-located users in a LOS, indoor environment with (``ON'') and without (``OFF'') the proposed MPC gain function extension. Both directive and omnidirectional antenna patterns are considered. At the BS side, a compact array with $M=32,128$ antennas is used. For comparison, plots of channels measured in the same environment and also shown (black, dotted line).}
  \label{fig_kappa_indoor}
\end{figure}
It turns out that the MPC gain function has quite a noticeable effect on the condition number~$\kappa(\vec{H}(t,f)),$ where $\vec{H}(t,f)$ is the $K\times M$ channel matrix at time~$t$ and frequency~$f$ of an MU-MIMO system with $K$ single-antenna users and $M$ BS antennas. In Fig.~\ref{fig_kappa_indoor}, a comparison of simulated MU-MIMO channels with and without the MPC gain function is presented. In particular, CDFs of the logarithmic channel condition number
\begin{align}
  \kappa_\op{dB}(\vec{H}(t,f)) = 20\log_{10}(\kappa(\vec{H}(t,f)))  
\end{align}
of $K=9$ closely-located, single-antenna users in a LOS, public auditorium~\cite[indoor site MS~1]{Flordelis:acc:18} are shown. We consider both directive and omnidirectional antenna patterns. Directive antenna patterns have been measured with an upper body phantom~\cite{HarryssonWCOM10} and in our simulations are oriented towards the auditorium stage. At the BS side, a compact cylindrical array with measured antenna patterns have been used. We consider $M=32$ and $128$ BS antennas. The positions and spreads of clusters are set as in the actual measured environment, while the rest of the COST \nolinebreak[1]2100 model parameters have been extracted from measured indoor channels~\cite{MAMMOET:D1.1:2015} and are reported in Table~\ref{tab_COST2100}. Simulation results are based on data from 10 runs. For each run, $B=257$ frequency points over 50~MHz of bandwidth, and $T=10$ positions (i.e., snapshots) over 0.25~m of a straight line, one for each user, have been simulated. For comparison, the CDFs obtained from measured channels are also included (black, dotted line). Prior to computing~$\kappa_\op{dB}$, the complex gains of both simulated and measured channels have been normalized by the expression
\begin{align}
  h_{km}(t,f) &= h_{km}^\op{meas}(t,f)\left(\frac{E_k}{MTB}\right)^{-1/2},\\
  E_k&=\sum\limits_{m=1}^M\sum\limits_{t=1}^T\sum\limits_{f=1}^B\abs{h_{km}^\op{meas}(t,f)}^2,
\end{align}
where $h_{km}^\op{meas}(t,f)$ are the unnormalized channels, so that the total channel gain for each user is~$MTB.$

As can be seen, when the MPC gain function is included in the simulations (``ON'' label), a good agreement is observed between simulated and measured channels for both 32 and 128 BS antennas. On the other hand, not including it (``OFF'' label) results in too large values of~$\kappa_\op{dB}$ by about 10~dB. This gap between the ``ON'' and ``OFF'' cases is reduced by 2~dB or so when measured rather than omnidirectional antenna patterns are applied.\footnote{This can be explained as follows. When the MPC gain function is turned ``OFF'' directive antennas can, by selecting different MPCs, help decorrelating the users' radio channels. When the gain function is turned ``ON'' users already ``see'' different MPCs and directive antenna patterns, by filtering away some of the MPCs, may have an adverse effect.} This improvement alone, however, is not sufficient to explain the observed low values of~$\kappa_\op{dB}$ measured, and the introduction of a MPC gain function seems to be required for realistic channel simulations.

To gain further insights into the behavior of the MPC gain function in multiuser systems, we next study the dependence of~$\kappa_\op{dB}$ on the angular spreads of the environment. For this we adopt a simple propagation model consisting of only one cluster whose MS- and BS-side angular spreads, denoted respectively by~$\Omega_\op{MS}$ and~$\Omega_\op{BS}$, can be controlled independently. (This is called a ``twin'' cluster in the terminology of the COST \nolinebreak[1]2100 model~\cite{VerdoneCOST12}). Fig.~\ref{fig_twin} shows the average
\begin{align}
  \bar{\kappa}_\op{dB} = \frac{1}{TB}\sum_{t=1}^T\sum_{f=1}^B \kappa_\op{dB}(\vec{H}(t,f))
\end{align}
of~$\kappa_\op{dB}$ for various MS- (left) and BS-side (right) angular spread values. Here,~$T=300,$ and for each $1\leq t\leq T$ the locations of the users have been randomly assigned inside a $2~$m radius circle. Furthermore, we have set the MPC-VR radius~$R_\op{MPC}=0.5,$ the effective number of MPCs~$N_\op{MPC}^\op{eff}=100,$ the number of users~$K=9,$ and the number of BS antennas~$M=128.$

In general, the channel condition number,~$\bar{\kappa}_\op{dB}$, decreases as the angular spreads increase. When the MPC gain function is turned ``ON'',~$\bar{\kappa}_\op{dB}$ always attains a lower value. Moreover, this value mainly depends on the BS angular spread~$\Omega_\op{BS}$, and not so much on the MS angular spread~$\Omega_\op{MS}$. When, on the other hand, the MPC gain function is turned ``OFF'', we see that~$\bar{\kappa}_\op{dB}$ depends on both~$\Omega_\op{BS}$ \emph{and}~$\Omega_\op{MS}$. This has interesting implications for the design and validation of MaMi systems. Most notably, the plots in Fig.~\ref{fig_twin} suggest that users having angular spreads as small as~$\Omega_\op{MS}=5^\circ$ can still be served concurrently provided that $\Omega_\op{BS}$ and~$M$ are sufficiently large. Conversely, if the decorrelating effects of the MPC gain function are not properly accounted for, results involving closely-located users may be overly pessimistic, even for large~$\Omega_\op{BS}$ and~$M$.
\begin{figure}[!t]
  \psfrag{MSangularspreaddeg}[][]{\tiny{$\Omega_\op{MS}$\,/\,degree}}
  \psfrag{BSangularspreaddeg}[][]{\tiny{$\Omega_\op{BS}$\,/\,degree}}
  \psfrag{AverageconditionnumberdB}[][]{\tiny{$\bar{\kappa}_\op{dB}$\,/\,dB}}
  \psfrag{numusers}[][]{\tiny{Number of users,~$K$}}
  \psfrag{deltakappa}[][]{\tiny{$\delta\bar{\kappa}_\op{dB}$\,/\,dB}}
  \psfrag{title1}[][]{\tiny{$\Omega_\op{BS}=60^\circ$}}
  \psfrag{title2}[][]{\tiny{$\Omega_\op{MS}=15^\circ$}}
  \psfrag{title3}[][]{\tiny{$\Omega_\op{BS}=60^\circ$}}
  \psfrag{ON}[][]{\tiny{ON}}
  \psfrag{OFF}[][]{\tiny{OFF}}
  \psfrag{Omegadeg15}[][]{\tiny{$\Omega_\op{MS}=15^\circ$}}
  \psfrag{Omegadeg60}[][]{\tiny{$\Omega_\op{MS}=60^\circ$}}
  \psfrag{Measured}[][]{\scriptsize{Directive}}
  \psfrag{Omni}[][]{\scriptsize{\makebox[.2cm]{}Omni.}}
  \centering
  \includegraphics[width=.65\textwidth]{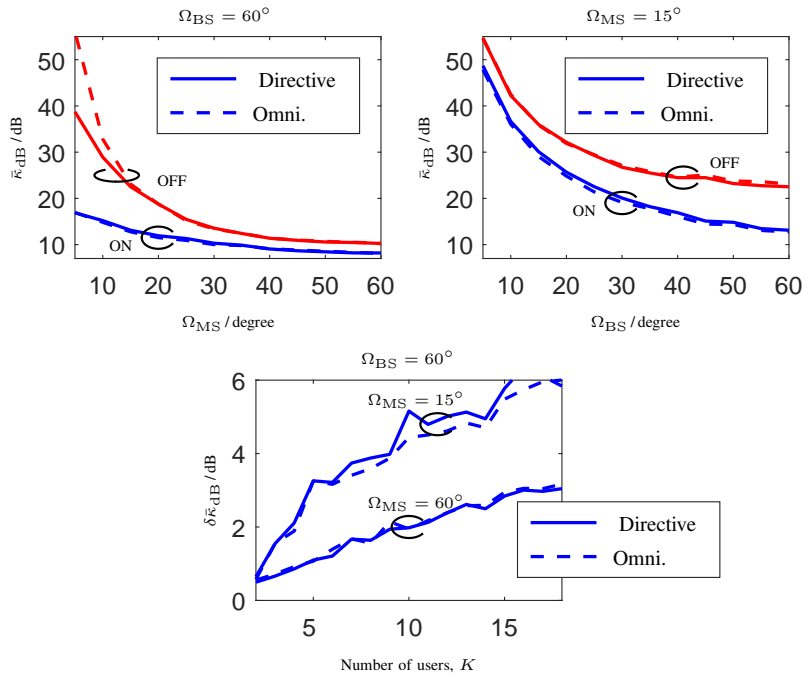}
  \caption{Average~$\bar{\kappa}_\op{dB}$ of the channel condition number of $K=9$ single-antenna MSs and a $M=128$-antennas BS connected by a single ``twin'' cluster. The MS-side (left) and the BS-side (right) cluster angular spreads, $\Omega_\op{MS}$ and $\Omega_\op{BS},$ are swept independently for the cases of ``ON''/``OFF'' MPC gain function, and directive\,/\,omnidirectional user antennas ($N_\op{MPC}^\op{eff}=100$).}
  \label{fig_twin}
\end{figure}
The significance of the MPC gain function will depend on the spatial density of concurrently served users, that is, the number of users communicating over the same time-frequency resource in a certain limited region. In the bottom pane of Fig~\ref{fig_twin}, plots of the gap
\begin{align}
  \delta\bar{\kappa}_\op{dB}=\bar{\kappa}_\op{dB,OFF}-\bar{\kappa}_\op{dB,ON}  
\end{align}
between the ``ON'' and ``OFF'' states are presented for~$\Omega_\op{BS}=60^\circ,$ $\Omega_\op{MS}=15^\circ,30^\circ,60^\circ,$ and~$K=2,\ldots,18.$ As can be seen, the gain function has a negligible effect for low values of~$K$, but its impact becomes noticeable as $K$ increases, in good agreement with our comments in Sec.~\ref{sec_acf_mpc}.

\begin{table*}[!ht]
  \renewcommand{\arraystretch}{.85}
  \scriptsize
  \centering
  \setlength\extrarowheight{-1.5pt}
  \caption{Parametrization of the COST 2100 model extension for closely-located users with PLA and CLA at~2.6~Ghz.}
  \label{tab_COST2100}
  \begin{threeparttable}
    \begin{tabular}{lrr}
      \hline
      Parameter&Outdoor VLA (NLOS)\tnote{1}&Indoor CLA (LOS)\\[-1.5pt]
      \hline
      Length of BS-VRs&&\\[-1.5pt]
      \quad$L_\op{BS}$~[m]&$3.2$&-\\[-1.5pt]
      Slope of BS-VR gain&&\\[-1.5pt]
      \quad$\mu_\op{BS}$~[dB/m]&0&-\\[-1.5pt]
      \quad$\sigma_\op{BS}$~[dB/m]&0.9&-\\[-1.5pt]
      MPC gain function&&\\[-1.5pt]
      \quad$\mu_{R_\op{MPC}}$~[dB]&-&-19.8\\[-1.5pt]
      \quad$\sigma_{R_\op{MPC}}$~[dB]&-&10.1\\[-1.5pt]
      Average number of visible far clusters&&\\[-1.5pt]
      \quad$N_\op{C}$&2.9$\times(L_\op{BS}+L)$&15\\[-1.5pt]
      Radius of the cluster visibility region&&\\[-1.5pt]
      \quad$R_\op{C}$~[m]&10&5\\[-1.5pt]
      Radius of cluster transition region&&\\[-1.5pt]
      \quad$T_\op{C}$~[m]&2&0.5\\[-1.5pt]
      Number of MPCs per cluster&&\\[-1.5pt]
      \quad$N_\op{MPC}$&31&1000\tnote{2}\\[-1.5pt]
      Cluster power decay factor&&\\[-1.5pt]
      \quad$k_\tau$~[dB/$\mu$s]&43&31\\[-1.5pt]
      Cluster cut-off delay&&\\[-1.5pt]
      \quad$\tau_\op{B}$~[$\mu$s]&0.91&0.25\\[-1.5pt]
      Cluster shadowing&&\\[-1.5pt]
      \quad$\sigma_\op{S}$~[dB]&7.6&2.7\\[-1.5pt]
      Cluster delay spread&&\\[-1.5pt]
      \quad$m_\tau$~[$\mu$s]&0.14&0.005\\[-1.5pt]
      \quad$S_\tau$~[dB]&2.85&1.5\\[-1.5pt]
      Cluster angular spread in azimuth (at BS)&&\\[-1.5pt]
      \quad$m_{\psi_\op{BS}}$~[deg]&7.0&4.6\\[-1.5pt]
      \quad$S_{\psi_\op{BS}}$~[dB]&2.4&2.1\\[-1.5pt]
      Cluster angular spread in elevation (at BS)&&\\[-1.5pt]
      \quad$m_{\theta_\op{BS}}$~[deg]&0&3.7\\[-1.5pt]
      \quad$S_{\theta_\op{BS}}$~[dB]&0&2.6\\[-1.5pt]
      Cluster angular spread in azimuth (at MS)&&\\[-1.5pt]
      \quad$m_{\psi_\op{MS}}$~[deg]&19&3.6\tnote{3}\\[-1.5pt]
      \quad$S_{\psi_\op{MS}}$~[dB]&2.0&2.1\tnote{3}\\[-1.5pt]
      Cluster angular spread in elevation (at MS)&&\\[-1.5pt]
      \quad$m_{\theta_\op{MS}}$~[deg]&0&0.7\tnote{3}\\[-1.5pt]
      \quad$S_{\theta_\op{MS}}$~[dB]&0&3.6\tnote{3}\\[-1.5pt]
      Cluster spread cross-correlation&&\\[-1.5pt]
      \quad$\rho_{\sigma_\op{S}\tau}$&-0.09&-0.45\\[-1.5pt]
      \quad$\rho_{\sigma_\op{S}\psi_\op{BS}}$&0.04&-0.56\\[-1.5pt]
      \quad$\rho_{\sigma_\op{S}\theta_\op{BS}}$&0&-0.50\\[-1.5pt]
      \quad$\rho_{\tau\psi_\op{BS}}$&0.42&0.70\\[-1.5pt]
      \quad$\rho_{\tau\theta_\op{BS}}$&0&0.34\\[-1.5pt]
      \quad$\rho_{\psi_\op{BS}\theta_\op{BS}}$&0&0.50\\[-1.5pt]
      Radius of LOS visibility region&&\\[-1.5pt]
      \quad$R_\op{L}$~[m]&-&30\tnote{3}\\[-1.5pt]
      Radius of LOS transition region&&\\[-1.5pt]
      \quad$T_\op{L}$~[m]&-&0\tnote{3}\\[-1.5pt]
      LOS power factor&&\\[-1.5pt]
      \quad$\mu_{K_\op{LOS}}$~[dB]&-&-5.2\\[-1.5pt]
      \quad$\sigma_{K_\op{LOS}}$~[dB]&-&2.9\\[-1.5pt]
      XPR&&\\[-1.5pt]
      \quad$\mu_{\op{XPR}}$~[dB]&0&9\tnote{4}\\[-1.5pt]
      \quad$\sigma_{\op{XPR}}$~[dB]&0&3\tnote{4}\\[-1.5pt]
      \hline
    \end{tabular}
    \begin{tablenotes}
      \item[1]{Reused from~\cite{Gao:2015:Extension}, except~$L_\op{BS},$ $N_\op{C}$.}
      \item[2]{$N_\op{MPC}^\op{eff}=10$.}
      \item[3]{Adopted from COST \nolinebreak[1]2100 channel model, 5~GHz indoor hall scenario~\cite{Kolmonen:2010:dual,Poutanen:2011:IndoorHall}.}
      \item[4]{Adopted from WINNER II channel models~\cite{WINNERII:2008}.}
    \end{tablenotes}
  \end{threeparttable}
\end{table*}

\section{Summary and Conclusions}Certain important aspects of the MaMi propagation channel are not captured by prevalent channel models. In this paper, we extend the conventional COST \nolinebreak[1]2100 model with the concepts of (1) BS-VRs, to model the appearance and disapperance of clusters along the axis of a physically-large array; and (2) MPC-VRs and MPC gain function, to model, at the MS side, birth-death processes of individual MPCs. Based on measurements of MaMi channels, the parameters of the proposed extensions have been statistically characterized, and their impact on the properties of MaMi propagation channels investigated theoretically and by simulations. In addition, it has been shown that simulation results are consistent with measurements of closely-located users, which demonstrates that the proposed extensions are able to capture the intended MaMi characteristics. A MATLAB implementation of the COST \nolinebreak[1]2100 model with the proposed extensions (as well as full support of other essential, but not discussed in this paper, MaMi characteristics such as spherical wavefronts and cluster dispersion in elevation) is freely available at~\cite{COST2100:GitHub}. Most importantly, this implementation provides a complete simulation framework for studies of MaMi channels for 5G and beyond.

\section*{Acknowledgment}
The presented investigations are based on data obtained in measurement campaigns performed by Xiang Gao, Fredrik Tufvesson, Ove Edfors, Tommy Hult, and Meifang Zhu, and Sohail Payami. We greatly benefited from discussions with Fredrik Rusek. This research has been conducted under the auspices of the COST-IRACON action.


\end{document}